\documentclass{amsart} 

\usepackage{amsmath,amsthm}     
\usepackage{graphicx}     
\usepackage{hyperref} 
\usepackage{url}
\usepackage{amsfonts} 

\usepackage{multicol}
\usepackage{comment}

\theoremstyle{theorem}

\newtheorem{prop}{Proposition}

\theoremstyle{definition}
\newtheorem*{definition}{Definition}

\newtheorem{example}{Example}
\allowdisplaybreaks

\makeatletter
\@addtoreset{footnote}{page}
\makeatother

\begin{document}

\title{Single Transferable Vote and Paradoxes of Negative and Positive  Involvement}

\author{David McCune\\               
\scriptsize William Jewell College\\    
Liberty, MO 64068\\                
mccuned@william.jewell.edu}                      

\begin{abstract}
We analyze a type of voting paradox which we term an \emph{involvement paradox}, in which a candidate who loses an election could be made into a winner if more of the candidate's non-supporters participated in the election, or a winner could be made into a loser if more of the candidate's supporters participated. Such paradoxical outcomes are possible under the voting method of single transferable vote (STV), which is widely used for political elections throughout the world. We provide a worst-case analysis of involvement paradoxes under STV and show several interesting examples of these paradoxes from elections in Scotland.

\end{abstract}


\maketitle

\section{Introduction}

Imagine you were a candidate for a political office and you have just lost a close election. You ask your campaign manager, ``Is there any way I could have won?'' and the manager responds, ``You could have won if the turnout of voters who dislike you were increased.'' Would you assume you had misheard? How could it be possible that it is better for your electoral chances to increase the turnout of people who do not support you? Interestingly, when using the election procedure of single transferable vote (STV), a voting method widely used throughout the world, it is sometimes possible that a losing candidate could be turned into a winner if more of that candidate's non-supporters cast ballots. Similarly, it is sometimes possible that a winning candidate can be turned into a loser if we increase turnout among that candidate's supporters. Such outcomes are called \emph{paradoxes of negative} or \emph{positive involvement} and represent a bit of mathematical irrationality at the heart of the STV voting method. In this article we explore how such outcomes can occur, providing theoretical worst-case outcomes and examining several real-world examples from STV elections in Scotland.

\section{Preliminaries}

STV is  used for political elections in many jurisdictions internationally. For example, STV is used to elect Senators in the Australian Senate, members of Ireland's lower legislative house, and council members in Scottish local government elections. In the United States the method is used for federal elections in the states of Alaska and Maine and for municipal elections in cities such as Minneapolis, Oakland, and San Francisco. The method is used for single-winner elections such as electing a mayor and also for multiwinner elections such as electing multiple members of a city council. A variety of similar voting rules can be classified as \emph{STV}; we examine what is considered the ``standard'' version of STV, which is the method used in Scotland. We focus on Scotland throughout because Scottish local elections are the greatest source of publicly available raw ballot data for STV elections. Most jurisdictions which use STV do not  provide ballot data, making the kind of real-world analysis in this article impossible. The Scottish elections are multiwinner (usually three or four candidates are elected at a time) and thus the multiwinner STV setting is our focus.


We now turn to defining and illustrating STV, and begin with notation. Let $n$ denote the number of candidates running in an election and $S$ denote the size of the winner set, which in our context is the number of legislative seats available. We usually denote the candidates $A$, $B$, $C$, etc., unless we are considering a real-world election, in which case we use actual candidate names. In an STV election, each voter casts a ballot which expresses a preference ranking of the candidates; when describing ballots, we use the symbol $\succ$ to denote that a voter prefers one candidate to another. For example, if $n=4$ then a voter could express that they rank $C$ first, $B$ second, $D$ third, and $A$ last, casting the ballot $C \succ B \succ D \succ A$. In most real-world STV elections voters are not required to provide a complete ranking of the $n$ candidates, so a voter could cast the ballot $D \succ B$, for example. Once all ballots are cast they are aggregated into a \emph{preference profile}, which shows how many voters cast each type of ballot. An example of a preference profile with 1804 voters and three candidates is given in Table \ref{profile}. The number 173 in the second column denotes that 173 voters cast the ballot $F \succ G\succ M$; the other numbers communicate similar information about the number of voters who cast the corresponding ballot in that column. We reserve the letter $P$ to denote a preference profile. 

It is difficult to provide a complete definition of STV in a concise fashion and thus we provide a high level description, accompanied by examples. A complete description of the Scottish STV rules can be found at \url{https://www.opavote.com/methods/scottish-stv-rules}. The STV method takes as input a pair $(P,S)$, which we call an \emph{election}, and outputs a set of winners $W(P,S)$ of size $S$ (we avoid the issue of ties; a careful general treatment of STV allows for the method to output multiple winner sets).  To calculate the winner set, STV works as follows: the method proceeds in rounds where in a given round either a candidate is eliminated from the election because they received too few first-place votes, or a candidate is elected to fill one of the $S$ seats because they received enough first-place votes. To earn a seat, a candidate's first-place vote total must reach the election's \emph{quota}, which is defined by \[\text{quota } = \left\lfloor \frac{\text{Number of Voters}}{S+1}\right\rfloor +1.\]

In a given round, if no candidate's first-place vote total achieves quota  then the candidate with the fewest first-place votes is eliminated, and this candidate's votes are transferred to the next candidate on their ballots who has not been elected or eliminated. If more than one candidate has the fewest first-place votes, one of these candidates is chosen to be eliminated through a tie-breaking mechanism, presumably at random. If a candidate's first-place vote total is greater than or equal to the quota, that candidate is elected and the votes they receive above quota (the candidate's \emph{surplus votes}) are transferred proportionally to the next non-eliminated and non-elected candidate which appears on the ballots being transferred. That is, the first-place votes that $A$ has earned in order to reach quota are ``locked in'' for $A$, and only $A$'s surplus votes are transferred to other candidates. The method continues in this fashion until $S$ candidates are elected, or until some number $S'<S$ of candidates have been elected by surpassing quota and there are only $S-S'$ candidates remaining who have not been elected or eliminated.  We fully illustrate this vote transfer process in the examples below. 

Our first example comes from a single-winner election in the Isle of Bute Ward of the Argyll and Bute Council area in Scotland. For purposes of local government, Scotland is divided into 32 council areas, each of which is subdivided into wards. Once every five years each ward elects a set number (usually three or four) of councillors using STV. The governing council for the council area consists of the winners from all of the area's wards.  If a councillor leaves the council for any reason, a by-election in that councillor's ward is held to fill the empty seat. The example below is one such single-winner by-election. 

The election offices of Scottish council areas provide preference profiles for their ward elections. STV has been used for Scottish local government elections since 2007, although most profiles from the 2007 elections are unavailable. The author collected profiles from the 2012, 2017, and 2022 elections, as well as several off-cycle by-elections, for \cite{MG}; the ballot data is available at \url{https://github.com/mggg/scot-elex}. The dataset also includes preference profiles for 21 elections from 2007 which were obtained from preflib.org \cite{Preflib}.


 \begin{example}\label{first-ex}\textbf{Argyll and Bute, 2021 by-election, Isle of Bute Ward.}
 The election contained the five candidates Findlay, Gillies, MacDonald, McCabe, and Wallace. The preference profile for this election is too large to display, and thus the top of Table \ref{profile} shows the profile after eliminating MacDonald and Wallace, the first two candidates in the actual election to be eliminated and have their votes transferred. To make the table more streamlined, we combine ballots of the form $A\succ B \succ C$ and $A\succ B$, as these ballots convey the same ranking information.
 
 \begin{table}[htb]
  \centering

\begin{tabular}{l|c|c|c|c|c|c|c|c|c}
Num. Voters & 173 & 157 & 371 & 66 & 208 & 265 & 87 & 210 & 267\\
\hline
1st choice & $F$ & $F$ & $F$ & $G$ & $G$& $G$& $M$ & $M$ &$M$\\
2nd choice & $G$ & $M$ &     & $F$ & $M$ &  & $F$ & $G$ & \\
3rd choice  & $M$& $G$ &  & $M$ & $F$ & &$G$&$F$&\\
\end{tabular}

\caption{The preference profile $P_1$ from the 2021 by-election in the Isle of Bute ward of the Argyll and Bute council area, after eliminating the bottom two candidates. The remaining candidates are Findlay (F), Gillies (G), and McCabe (M).}
  \label{profile}
  \end{table}

In this election $S=1$, and in the single-winner case quota is set at a majority of votes. There were 2013 total votes cast and thus quota is $\lfloor 2013/2 \rfloor +1 = 1007$. (Note that the votes in the preference profile do not sum to 2013 because we have already eliminated two candidates, and thus have eliminated ballots in which only those two candidates are ranked.)  As can be seen from the preference profile, Findlay initially has $173+157+371=701$ first-place votes, Gillies has 539, and McCabe has 564. No candidate achieves quota and thus Gillies is eliminated. As a result, 66 votes are transferred to Findlay and 208 to McCabe, and McCabe defeats Findlay 772 to 767. Because $772<1007$, the formal STV algorithm eliminates Findlay and according to the preference profile McCabe's final vote total is 1168.

In most real-world elections the preference profile is too large to display. Consequently, election offices display how an election unfolds using a \emph{votes-by-round table}, which shows the number of first-place votes controlled by each candidate in each round. The votes-by-round table for this election is shown in the left of Table \ref{votes_IRV}. When a candidate is elected their vote total is made bold, and when a candidate is elected or eliminated their votes are not shown for subsequent rounds.

\begin{table}[htb]
  \centering
  
  \begin{tabular}{ccc}
  
  \begin{tabular}{c|ccc}
  \hline
  Cand. & \multicolumn{3}{c}{Votes by Round}\\
  \hline
  $F$ & 701 & 767 &  \\
  $G$ & 539 &          &           \\
  $M$ &  564& 772 & \textbf{1168}\\
    \hline
  \end{tabular}
  
  & &
  
  \begin{tabular}{c|ccc}
  \hline
  Cand. & \multicolumn{3}{c}{Votes by Round}\\
  \hline
  $F$ & 701 & 788 &\textbf{1288}  \\
  $G$ & 565 & 775         &           \\
  $M$ &  564&  &\\
    \hline
  \end{tabular}
  
  \end{tabular}

\caption{(Left) The votes-by-round table for the preference profile in Table \ref{profile}. (Right) The resulting votes-by-round table when adding 26 ballots on which Gillies is ranked first and Findlay is ranked last.}
  \label{votes_IRV}
  \end{table}

\end{example}

In the previous example, suppose we add 26 ballots of the form \begin{center}Gillies$\succ$McCabe$\succ$MacDonald$\succ$Wallace$\succ$Findlay \end{center} to the ballots in the preference profile in Table \ref{profile}. That is, suppose we can find 26 voters with the preference ranking above who did not participate in the election, and we imagine what would have happened if they had cast a ballot. What effect should this have on the electoral outcome? Since these ballots rank Gillies first, it would not be unreasonable for Gillies to win the resulting election, as we have increased his electoral support. Arguably McCabe should still win because she won the original election and these voters seem generally to support her. These seem like the only two reasonable outcomes. However, the reader can check that the addition of these ballots would cause Findlay to win. Surprisingly, adding ballots on which Findlay is ranked last can create a better outcome for her.

The right of Table \ref{votes_IRV} shows how Findlay can win despite the addition of these ballots. In the resulting election Gillies has more first-place votes than McCabe and thus McCabe is eliminated first and Findlay faces Gilles head-to-head in the penultimate round. Even with the additional support of these 26 ballots Gilles cannot defeat Findlay, who now wins the seat. This outcome is an example of a \emph{paradox of negative involvement}, in which there exists a set of  ballots with a candidate $C$ ranked last such that the addition of these ballots to the original ballot data turns $C$ from a loser into a winner.

\begin{definition}
An election $(P,S)$ demonstrates a \textbf{paradox of negative involvement} if there exists a candidate $L\not\in W(P,S)$ and a collection of identical ballots $\mathcal{B}$ with $L$ ranked in last place such that if we add the ballots from $\mathcal{B}$ to the ballots in $P$, creating a new profile $P'$, then $L \in W(P',S)$.
\end{definition}

We use the term \emph{paradox of involvement} following \cite{Duddy}. There is no universally agreed upon language for these paradoxes;  work such as \cite{LM} uses the term ``negative participation paradox'' to describe this same kind of paradox, for example. Regardless of terminology, as far as we know these kinds of paradoxes have previously been studied only in the single-winner case. The notion of an involvement paradox is closely related to the concepts of a no-show paradox \cite{BF, FT, Moulin}, monotonicity paradox \cite{MG, FT, EFSS, ON}, and a paradox of positive responsiveness \cite{May}.

When presented with a paradox such as what is demonstrated in Example \ref{first-ex}, a natural response is that this kind of hypothetical thought experiment is not important. We do not know if there actually exist 16 people in the general electorate of the Isle and Bute Ward who did not participate in the election but, if they had, would have cast a ballot with Gillies ranked first and Findlay ranked last. Thus, why should we care about the concept of a negative involvement paradox?  A possible answer is that one way to evaluate whether a voting method is ``good'' is to analyze if that method behaves reasonably in response to changes in the ballot data. If we make systematic changes to the ballots and observe a paradoxical outcome when using a given voting method, perhaps there is too much mathematical irrationality baked into the design of the method for it to be used. Analyzing negative involvement paradoxes is one way to illustrate that STV sometimes behaves in absurd ways when we perturb the ballot data, which perhaps implies (from a mathematical perspective) that STV crosses a kind of ``reasonableness'' red line.

Our next example demonstrates how STV works if $S>1$, and shows how a paradox of negative involvement can arise in this multiwinner context. We provide a hypothetical example so that we can give a complete preference profile.

\begin{example}\label{second-ex}
Consider the preference profile $P_2$ in Table \ref{second_profile}. There are 10,000 voters and thus in the election $(P_2, 2)$, quota is $\lfloor10000/3\rfloor+1=3334$ votes. Candidate $A$ receives 3350 first-place votes, thereby achieving quota in the first round and winning the first seat. $A$ surpasses quota by only 16 votes, which must be transferred. Candidate $B$ earns $(100/3350)\cdot 16=0.48$ of these votes while $C$ earns $(3250/3350)\cdot 16=15.52$, as shown in the left table of Table \ref{second_ex_votes}. Note that we transfer fractions of votes, which is common under STV because surplus votes are transferred proportionally. After $A$'s votes are transferred no candidate achieves quota, and thus $D$ is eliminated. As a result, 1089 of their votes are transferred to $B$ (candidate $A$ has already been elected and is ineligible for vote transfers) and 1091 transferred to $C$, and $B$ achieves quota and wins the second seat. Thus $W(P_2,2)=\{A,B\}$. See the left of Table \ref{second_ex_votes} for the votes-by-round table.

If the election contained more candidates or the number of seats were larger then this process of elimination and vote transfers can take many more rounds, but all STV elections (under the Scottish rules, which seem mostly standard internationally for political elections) unfold in this manner, up to some minor technical details.

What if an additional 100 voters decided to cast their ballots, and all of these voters do not like $C$? In particular, what should happen if we add 100 ballots of the form $A\succ B \succ D \succ C$ to the ballots in $P_2$, creating the profile $P_2'$? Since $A$ and $B$ were the winners of the original election $(P_2,2)$ and these additional voters' preferred winner set is (presumably) $\{A,B\}$, the reasonable answer is that $W(P_2',2)$ should be $\{A,B\}$. At the very least, $C$ should not benefit from the participation of extra voters who rank this candidate last. However, as shown in the right table of Table \ref{second_ex_votes}, the addition of these ballots cause $C$ to win a seat, demonstrating a paradox of negative involvement. Most voters in $P_2$ who rank $A$ first also rank $C$ second, and by adding these ballots we allow $A$ to surpass quota by a much larger margin, thereby creating a much larger flow of transfer votes to $C$. 

The presence of this paradox is perhaps unsurprising given Example \ref{first-ex}, but note that in the multiwinner case this paradox can arise in different ways than the single-winner case. By adding these ballots we do not change the order of elimination or election in previous rounds; whether we add these ballots or not, $A$ is elected in the first round, $D$ is eliminated in the second round, and $B$ and $C$  face off in the third round. However, with the addition of these ballots $C$ wins over $B$ by a comfortable margin, the reverse of the situation without the additional ballots. 

We note that we engineered a paradox of negative involvement by adding 100 ballots, but we could add anywhere from 22 to 3,062 ballots of the form $A\succ B \succ D \succ C$ and $C$ would win a seat in the resulting election. It is surprising that we can increase the size of the electorate by 30\%, using thousands of additional voters who would prefer anyone but $C$ be on the council, and $C$ benefits.
\end{example}

\begin{table}[htb]
  \centering

\begin{tabular}{l|c|c|c|c|c|c}
Num. Voters & 100 & 3250 & 2250 & 2220 & 1089&1091 \\
\hline
1st choice & $A$ & $A$ & $B$  & $C$ &$D$ & $D$\\
2nd choice & $B$ & $C$ & $D$  & $A$& $A$& $C$\\
3rd choice  & $C$& $B$ &  $A$ & $B$ &$B$& $B$\\
4th choice   & $D$&$D$ & $C$ & $D$ &$C$ & $A$\\
\end{tabular}

\caption{A preference profile $P_2$ with four candidates.}
  \label{second_profile}
  \end{table}

\begin{table}[htb]
  \centering
  
  \begin{tabular}{ccc}
  
  \begin{tabular}{c|ccc}
  \multicolumn{4}{c}{Quota $=3334$}\\
  \hline
  Cand. & \multicolumn{3}{c}{Votes by Round}\\
  \hline
  $A$ & \textbf{3350}& &  \\
  $B$ & 2250 &   2250.48       & \textbf{3339.48}          \\
  $C$ &  2220& 2235.52 & 3326.52\\
  $D$ & 2180 & 2180 & \\
    \hline
  \end{tabular}
  
  & &
  
  \begin{tabular}{c|ccc}
  \multicolumn{4}{c}{Quota $=3367$}\\
  \hline
  Cand. & \multicolumn{3}{c}{Votes by Round}\\
  \hline
  $A$ & \textbf{3450}& &  \\
  $B$ & 2250 &   2254.81       &  3343.81     \\
  $C$ &  2220& 2298.18 & \textbf{3389.19}\\
  $D$ & 2180 & 2180 & \\
    \hline
  \end{tabular}
  
 \end{tabular}
\caption{(Left) The votes-by-round table for the preference profile in Table \ref{profile} when using STV with $S=2$. (Right) The resulting votes-by-round table when adding 100 ballots of the form $A\succ B\succ D \succ C$.}
  \label{second_ex_votes}
  \end{table}
 
We close this section by noting that we can also define a paradox of \emph{positive} involvement \cite{Duddy}, which cannot be demonstrated by STV when $S=1$. We provide examples of this paradox in the multiwinner setting in the following sections.
    
 \begin{definition}
An election $(P,S)$ demonstrates a \textbf{paradox of positive involvement} if there exists a winning candidate $C$ and a collection of ballots $\mathcal{B}$ with $C$ ranked first  such that if we add $\mathcal{B}$ to the ballots in $P$, creating a new profile $P'$, then $C \not\in W(P',S)$.
\end{definition}

\section{Worst-Case Scenarios}\label{worst-case-section}

If $S=1$ then it is straightforward to show that we cannot demonstrate a paradox of involvement by adding ballots on which the original winner is ranked first. In particular, it is not possible to add ballots with the winner $W$ ranked first and candidate $L$ ranked last such that the addition of these ballots turns $L$ into the winner. The most extreme version of an involvement paradox in the single-winner case is illustrated by Example \ref{first-ex}, where we add ballots with the winner ranked second and the last-ranked candidate on these ballots becomes the winner.

This raises the question: in the multiwinner setting, how extreme can these paradoxes be? Is it possible to add ballots on which the original winners are all ranked in the top $S$ slots and the addition of these ballots turns all winners into losers? As we show below, the answer is No (Proposition \ref{first-prop}), but almost (Proposition \ref{worst-case-prop}). 

\begin{prop}\label{first-prop}
Let $(P,S)$ be an election and let $W(P,S)=\{C_1, \dots, C_S\}$. Suppose we add a set of identical ballots to the ballots in $P$, creating a new preference profile $P'$, such that $C_1, \dots, C_S$ are ranked in the top $S$ slots on the added ballots. Then $W(P,S)\cap W(P',S)\neq \emptyset$.

\end{prop}

\begin{proof}
Note the proposition is trivially true if $n<2S$ because $W(P,S)$ and $W(P',S)$ have size $S$, and if $n<2S$ then the two sets must have at least one candidate in common. Thus we assume $n\ge 2S$.

Let $W(P,S)=\{C_1, \dots, C_S\}$. Let $\mathcal{B}$ be a collection of identical ballots such that the candidates from $W(P,S)$ are ranked in the top $S$ rankings, and let $P'$ be the preference profile obtained by adding the ballots in $\mathcal{B}$ to the ballots from $P$. Due to the presence of partial ballots, it is possible that every candidate in $(P,S)$ earns a seat without achieving quota. In this case, the candidate ranked at the top of the ballots in $\mathcal{B}$ will earn a seat in $(P',S)$. 

Assume at least one candidate in $W(P,S)$ wins a seat in the election $(P,S)$ by achieving quota. Let $C_1$ be the winning candidate who earns the first seat in $(P,S)$. That is, $C_1$ achieves quota before any other winning candidate or has the highest vote total in the first round in which multiple candidates are elected.  Suppose $C_j \not \in W(P',S)$ for all $C_j \in W(P,S)$, $j\neq 1$. We claim $C_1 \in W(P',S)$.

Assume $C_1 \not \in W(P',S)$ and let $W(P',S)=\{C_{S+1}, \dots C_{2S}\}$. Let $V_P$ denote the number of voters in $(P,S)$ and let $k$ denote the round in which $C_1$ achieves quota in $(P,S)$. Note that in $(P',S)$ the same candidates are eliminated prior to round $k$ as were eliminated in $(P,S)$. To see why, let $C_r$ denote the candidate ranked at the top of the ballots in $\mathcal{B}$. In moving from $P$ to $P'$, the only difference in how the elections unfold up to round $k$ is that $C_r$ has more votes, but otherwise the same candidates are eliminated in each round. It is perhaps \emph{a priori} possible that $C_r$ is elected prior to round $k$ in $P'$ because of their additional support from the ballots in $\mathcal{B}$, but this cannot occur because $C_r$ does not achieve quota in $(P',S)$ by assumption. Thus we assume WLOG that $k=1$ and $C_1$ earns at least $\lfloor V_P/(S+1)\rfloor +1$ votes initially in $(P,S)$. 

Since $C_1 \not \in W(P',S)$, either $C_{S+1}, \dots, C_{2S}$ achieve quota in $(P',S)$ before $C_1$ is eliminated or there exists a round in which $C_1$ has the fewest first place votes and none of $C_{S+1}, \dots, C_{2S}$ have been eliminated. In the former case, in the round in which the winners achieve quota at least $\lfloor V_P/(S+1)\rfloor +1$ votes are controlled by $C_1$ and none of the ballots in $\mathcal{B}$ have been transferred to any candidate in $W(P',S)$. (Note that even if $C_1$ is not ranked at the top of the ballots in $\mathcal{B}$, $C_1$ would need to be eliminated before these ballots could be transferred to any candidates in $W(P',S)$.) For each winner in $W(P',S)$ to achieve quota while $C_1$ does not, each of these winners needs to have at least $\lfloor V_P/(S+1)\rfloor +2$ first-place votes (note that, since $C_1 \not \in W(P',S)$, quota must be greater than $\lfloor V_P/(S+1)\rfloor +1$). But it is straightforward to show that \[\lfloor V_P/(S+1)\rfloor +1 + S(\lfloor V_P/(S+1)\rfloor +2)+|B| > V_P + |B|.\] That is, there are not enough votes to go around in $(P',S)$ for the new winners to achieve quota if $C_1$ has not been eliminated. 

On the other hand, suppose there is a round in which $C_1$ is eliminated and none of $C_{S+1}, \dots, C_{2S}$ have been eliminated. For this to occur, there must exist a round in which $C_1$ has the fewest first-place votes. In this round, $C_1$ has at least $\lfloor V_P/(S+1)\rfloor +1$ votes and the ballots in $\mathcal{B}$ have not been transferred to any candidates in $W(P',S)$. Thus the first-place vote totals for each of $C_{S+1}, \dots, C_{2S}$ must surpass $\lfloor V_P/(S+1)\rfloor +1$ without any assistance from the ballots in $\mathcal{B}$. As in the previous case, there are not enough ballots to go around.\end{proof}

Proposition \ref{first-prop} rules out the most extreme potential example of an involvement paradox under STV, that we can add ballots such that the bottom $S$ ranked candidates replace the top $S$ ranked candidates in the winner set. The next proposition shows it is possible to add ballots so that the bottom $S-1$ ranked candidates on these ballots replace the top $S-1$ ranked candidates in the winner set. 

\begin{prop}\label{worst-case-prop}

Let $S>1$. There exists a preference profile $P$ and a set of identical ballots $\mathcal{B}$ such that

\begin{enumerate}
\item the candidates in $W(P,S)$ are ranked in the top $S$ slots on the ballots in $\mathcal{B}$,
\item when the ballots in $\mathcal{B}$ are added to the ballots in $P$, creating the preference profile $P'$, then $|W(P,S)\cap W(P',S)|=1$ and the top $S-1$ ranked candidates on the ballots in $\mathcal{B}$ are not in $W(P',S)$,
\item if $C$ is ranked in the bottom $S-1$ rankings on the ballots in $\mathcal{B}$ then $C\in W(P',S)$.
\end{enumerate}

\end{prop}

To prove this proposition then for each $S>1$ we must construct a profile with the desired three properties. To motivate such a construction, Table \ref{prop-proof-ex} provides a profile $P$ with five candidates which demonstrates the above properties for $S=3$. The profile contains 4000 voters and thus quota in the election $(P,3)$ is 1001. The votes-by-round table in the bottom left of Table  \ref{prop-proof-ex} shows how the election unfolds. Initially $C_2$ has the fewest first-place votes but $C_1$ surpasses quota in round 1 by two votes, both of which are transferred to $C_2$. As a result $C_2$ no longer has the fewest first-place votes, and $C_4$ is eliminated in round 2. Their surplus votes are transferred to $C_2$ and $C_3$, and $W(P,3)=\{C_1, C_2, C_3\}$.

If we add 8 ballots of the form $C_3\succ C_2\succ C_1\succ C_4\succ C_5$ to create a new profile $P'$ then $C_1$ still achieves quota in the first round but no longer has any surplus to transfer. As a result $C_2$ is eliminated in round 2, their votes are transferred to $C_4$ and $C_5$, and $W(P',3)=\{C_1, C_4, C_5\}$. This is an example of a paradox of negative and positive involvement since the ballots added rank $C_3$ at the top and $C_5$ at the bottom, and the addition of these ballots turns $C_3$ into a loser and $C_5$ into a winner. 

\begin{table}[htb]
  \centering

\begin{tabular}{l|c|c|c|c|c|c|c|c|c}
Num. Voters & 1003&373&374&375&376&374&374&375&376\\
\hline
1st choice & $C_1$&$C_2$&$C_2$&$C_3$&$C_3$&$C_4$&$C_4$&$C_5$&$C_5$\\
2nd choice & $C_2$&$C_4$&$C_5$&$C_4$&$C_5$&$C_2$&$C_3$&$C_2$&$C_3$\\
\vdots  & $\vdots$ &$\vdots$ &$\vdots$ &$\vdots$ &$\vdots$ &$\vdots$ &$\vdots$ &$\vdots$ &$\vdots$ \\

\end{tabular}
\vspace{.1 in}

\begin{tabular}{ccc}
  
  \begin{tabular}{c|cccc}
  \multicolumn{4}{c}{Quota $=1001$}\\
  \hline
  Cand. & \multicolumn{3}{c}{Votes by Round}\\
  \hline
  $C_1$ & \textbf{1003}&&   \\
  $C_2$ & 747 &   749       & \textbf{1123}&       \\
  $C_3$ &  751&751&\textbf{1125}\\
  $C_4$ & 748&748& \\
  $C_5$ & 751&751&751\\
    \hline
  \end{tabular}
  
  & &
  
  \begin{tabular}{c|cccc}
  \multicolumn{4}{c}{Quota $=1003$}\\
  \hline
  Cand. & \multicolumn{3}{c}{Votes by Round}\\
  \hline
  $C_1$ & \textbf{1003}&&   \\
  $C_2$ & 747 &   747       & &       \\
  $C_3$ &  759&759&759\\
  $C_4$ & 748&748& \textbf{1121}\\
  $C_5$ & 751&751&\textbf{1125}\\
    \hline
  \end{tabular}
  
 \end{tabular}

\caption{(Top) A preference profile $P$ which illustrates the proof of Proposition \ref{worst-case-prop}. (Bottom Left) The votes-by-round table for $P$. (Bottom Right) The votes-by-round table adding 8 ballots of the form $C_3\succ C_2\succ C_1\succ C_4\succ C_5$ to the ballots in $P$.}
  \label{prop-proof-ex}
  \end{table}

By Proposition \ref{first-prop}, this example represents the most extreme outcome of this form of paradox for $S=3$.  For larger values of $S$ we can generalize this construction by ensuring that $C_2$ has the fewest first-place votes and $C_{S+1}$ has the second-fewest. For each other candidate $C_j$, if $j\le S$ (respectively $j>S$) then we ensure that ballots on which $C_j$ is ranked first  have second-place preferences which are spread as evenly as possible among candidates $C_{S+1}, \dots, C_{2S-1}$ (respectively, $C_2, \dots, C_{S}$). If $C_1$ just barely surpasses quota and the first-place vote totals for all other candidates are approximately equal, as in the top of Table \ref{prop-proof-ex}, then we can exhibit paradoxes of negative and positive involvement in a manner similar to what occurs in the bottom of Table \ref{prop-proof-ex}.

\begin{proof}

In building a profile for a given $S$ we wish to minimize the number of candidates used. If $S>2$ then we can use $2S-1$ candidates, but it is straightforward to show that if $S=2$ then we cannot build a profile with the desired properties using three candidates. Thus, we separate the case of $S=2$ from $S>2$. 

For $S=2$, consider the preference profile $P$ below with 1000 voters. Note $W(P,2)=\{C_1,C_2\}$, as $C_1$ transfers no surplus votes in round 2 and the elimination of $C_4$ ensures a transfer of votes only to $C_2$. However, if we add 3 ballots of the form $C_2\succ C_1 \succ C_4\succ C_3$, creating a new profile $P'$, then the quota increases to 335. As a result $C_1$ does not achieve quota in the first round; instead $C_4$ is eliminated and 40 votes transferred to $C_1$. Now $C_1$ has a substantial surplus which is transferred to $C_3$, allowing $C_3$ to defeat $C_2$ in the penultimate round. Thus, $W(P',2)=\{C_1,C_3\}$ and we are done with the case $S=2$.

\begin{center}
\begin{tabular}{l|c|c|c|c}
Num. Voters & 334 & 314 & 312 & 40\\
\hline
1st choice & $C_1$ & $C_2$ & $C_3$ & $C_4$\\
2nd choice & $C_3$ & $C_1$ & $C_2$ & $C_1$\\
3rd choice & $C_4$ & $C_3$ & $C_1$ & $C_2$\\
4th choice & $C_2$ & $C_4$ & $C_4$ & $C_3$\\

\end{tabular}
\end{center}

Let $S>2$. We create a preference profile $P$ with $(S+1)10^S$ voters and $2S-1$ candidates, labeled $C_1, C_2, \dots, C_{2S-1}$. Note that quota is  $\lfloor (S+1)10^S/(S+1) \rfloor + 1 = 10^S+1$.  We now describe the types of ballots which make up the ballot set for $P$.

Let $V_1=10^S+3$ ballots have the form $C_1\succ C_2 \succ \dots$, where the third through $(2S-1)$st rankings can be filled in arbitrarily.

Let $V_2=\left\lfloor \frac{(S+1)10^S-10^S-3}{2S-2}\right\rfloor -2$ ballots have $C_2$ ranked at the top. Of these, let $1/(S-1)$ of them have have $C_j$ ranked second for $S+1\le j\le 2S-1$. Since $V_2/(S-1)$ may not be an integer, assign $\left\lfloor V_2/(S-1) \right\rfloor$ or $\left\lceil V_2/(S-1) \right\rceil$ voters to each of these type of ballot so that the sum of the voters who rank $C_2$ first is $V_2$ (the exact way the floor or ceiling is assigned to a type of ballot with $C_2$ ranked first is not important). The rankings on these ballots past the second ranking can be completed in an arbitrary manner.

Let $V_{S+1}=V_2+1$ ballots have $C_{S+1}$ ranked at the top. Of these, let $1/(S-1)$ of them have $C_j$ ranked second for $2\le j \le S$ and subdivide the $V_{S+1}$ ballots into the $S-1$ types in a similar manner as was done for $C_2$. For each type of ballot $C_{S+1}\succ C_j \succ \dots$, assign  $\left\lfloor V_{S+1}/(S-1) \right\rfloor$ or $\left\lceil V_{S+1}/(S-1) \right\rceil$ voters to this type of ballot in any way that the sum of the ballots with $C_{S+1}$ ranked at the top is $V_{S+1}$.

For any other candidate $C_j$, let $V_j \in \{V_2+2,V_2+3,V_2+4\}$ such that $\displaystyle\sum_{i=1}^{2S-1}V_i=(S+1)10^S$. As with other choices above, the exact choice of value for $V_j$ does not matter. Create $V_j$ ballots with $C_j$ ranked first. If $2\le j \le S$ then split the $V_j$ ballots into types as was done for $C_2$; if $S+1\le j \le 2S-1$ then split the $V_j$ ballots into types as was done for $C_{S+1}$.

We claim $W(P,S)=\{C_1, \dots, C_S\}$. To see this, note that by construction $C_1$ is the only candidate to achieve quota in round 1, and they have two surplus votes to transfer. These two votes go to $C_2$. In round 2 no candidate achieves quota and thus $C_{S+1}$ is eliminated. Their votes are transferred to $C_2, \dots, C_S$ in almost equal parts. The election continues in this manner with candidates from $\{C_j\}_{j\ge S+1}$ being eliminated and votes transferring to candidates in $\{C_j\}_{j\le S}$. At some point in the transfer process it is possible that some candidates in $\{C_j\}_{j\le S}$ achieve quota while others do not, which would cause a small transfer of votes to surviving candidates in $\{C_j\}_{j\ge S+1}$. However, by construction in each round of the election the votes for candidates $\{C_j\}_{j\le S}$ are approximately equal, and thus any surplus vote transfers from such an event would be minuscule in comparison to the vote transfers from $\{C_j\}_{j\le S}$ to $\{C_j\}_{j\le S}$ due to candidate elimination. Thus $W(P,S)=\{C_1, \dots, C_S\}$.

Suppose we add $2(S+1)$ ballots of the form \[C_3\succ C_4 \succ \dots \succ C_S\succ C_2 \succ C_1 \succ C_{S+1}\succ\dots \succ C_{2S-1}\] to the ballot data in $P$, creating a new preference profile $P'$. The addition of these ballots increases quota by two so that in the election $(P',S)$ candidate $C_1$ still achieves quota in the first round but now has no surplus votes to transfer. As a result $C_2$ is eliminated first, setting off a chain reaction of eliminations similar to what occurs in $(P,S)$. By construction, in this case the candidates who survive to the end are $\{C_1,C_{S+1}, \dots, C_{2S-1}\}$.\end{proof}

In practice, we would not expect to see such extreme examples of these paradoxes like those that occur in this section because real-world ballot data does not look like the profile in Table \ref{prop-proof-ex}. However, real-world elections can provide fun and interesting examples of these paradoxes, as we illustrate in the next section.

\section{Real-World Examples}\label{real-world-ex}

In this section we illustrate how paradoxes of negative and positive involvement can manifest in actual elections from Scotland. The dataset of Scottish local government elections contains 1100 elections in total and  we found 99 which demonstrate an involvement paradox. Only two demonstrate a paradox of positive involvement, while all 99 demonstrate a negative involvement paradox. We present four of the most interesting examples, focusing on the most extreme outcomes. A complete description of how these paradoxes arise in each of the 99 elections is provided in the appendix. The appendix also contain a description of how we searched the dataset for elections which demonstrate these paradoxes. 

When engineering a paradox we attempt to make it as ``paradoxical'' as possible. In our view, the strangest outcomes occur when the added ballots rank the original winners at the top of the ballots, since the addition of such ballots seemingly should only reinforce the victory of those candidates. Furthermore, the last-place candidate on the ballots must replace some winning candidate in the winner set; we try to rank this original winner as highly as possible on the added ballots. For all but one of the 99 elections, the paradox occurs by having the last-ranked candidate replace one other candidate in the winner set but all other winners retain their seats, so we do not observe outcomes nearly as extreme as those analyzed in the worst-case scenario. We could find only one election (Example \ref{argyll-bute-ex}) in which multiple winners lose their seats as a result of an involvement paradox.

Our first example comes from a ward in the East Dunbartonshire council area. This is a typical example of the kind of paradox of negative involvement that we see in the Scottish data, and is much in the spirit of Example  \ref{second-ex}.

\begin{example}\label{e-duns-ex}\textbf{East Dunbartonshire, 2022, Bishopbriggs North and Campsie Ward.} The votes-by-round table for the actual election is displayed in the top of Table \ref{e-duns-table}. Ferretti and McDiarmid achieve quota in the first round, and Hendry and Williamson achieve quota in the seventh. Note that Williamson defeats Pews in the penultimate round by about 170 votes.

Suppose we add 483 ballots of the form \begin{center}
Hendry$\succ$McDiarmid$\succ$Williamson$\succ$Ferretti$\succ$Gallacher$\succ$Harris$\succ$Rowan$\succ$Pews
\end{center} to the ballot data. Then the election would unfold as shown in the bottom of Table \ref{e-duns-table}, and Pews replaces Williamson in the winner set. By adding hundreds of ballots on which Pews is ranked last we cause Pews to win a seat.

How does this example work? Note from the votes-by-round table for the original election that Pews benefits more from the election of McDiarmid than Williamson does. After McDiarmid's surplus votes are transferred (see the vote totals for the third round) Pews has picked up an extra 19.9 votes while Williamson has picked up only 5.3, and thus McDiarmid's election benefits Pews over Williamson by a ratio of about $4:1$. Thus, it should be beneficial for Pews if we could inflate McDiarmid's victory. Furthermore, far more Hendry voters prefer Pews to Williamson, and thus an inflated win for Hendry would also benefit Pews. To justify this, we counted the ballots on which Hendry appears and another candidate is ranked immediately after Hendry. There are 1958 of these ballots, of which 801 have the property that Pews is ranked immediately after Hendry; the corresponding number drops to 47 for Williamson. In the actual election this disparity does not help Pews because Hendry is not elected until after Pews is eliminated. If we add lots of ballots on which Hendry is ranked first then we cause Hendry to achieve quota much earlier and Pews can benefit from the subsequent vote transfer. Note that in the resulting election (bottom of Table \ref{e-duns-table}), Hendry and McDiarmid both surpass quota by a significant margin, much to the benefit of Pews. Thus, this example operates much in the same fashion as Example \ref{second-ex}.

\begin{table}

\begin{tabular}{c | c | c|c|c|c|c|c}
\multicolumn{8}{c}{\textbf{Actual Election}, Quota $=1615$}\\
\hline
Candidate & \multicolumn{7}{c}{Votes by Round}\\
\hline
Ferretti & \textbf{2130} &&&&&&\\
Gallacher & 164 & 175.4 & 178.2  & 200.0 &&&\\
Harris & 121 & 135.5 & 136.5 &&&&\\
Hendry & 1495 & 1502.5 & 1509.3 & 1522.4 & 1578.8 & 1596.8 & \textbf{2161.9}\\
McDiarmid & \textbf{1674} &&&&&&\\
Pews & 1229 & 1254.1 & 1274.0 & 1289.1 & 1333.2 & 1429.8 & \\
Rowan & 410 & 462.5 & 468.7 & 491.7 & 508.3 &&\\
Williamson&849 & 1223.3&1228.6&1268.2&1292.7&1592.9&\textbf{1823.7}\\
\hline

\end{tabular}

{\small
\vspace{.1 in}
\begin{tabular}{c | c | c|c|c|c|c|c|c}
\multicolumn{9}{c}{\textbf{Modified Election}, Quota $=1712$}\\
\hline
Candidate & \multicolumn{8}{c}{Votes by Round}\\
\hline
Ferretti & \textbf{2130} &&&&&&&\\
Gallacher & 164 & 172.0 & 179.6  & 185.6 & 210.6&&&\\
Harris & 121 & 132.0 & 133.5 &135.4&&&&\\
Hendry & \textbf{1978} &&&&&&&\\
McDiarmid & 1674 &1711.1&\textbf{1826.1}&&&&&\\
Pews & 1229 & 1245.7 & 1326.0 & 1371.1 & 1391.9 & 1464.9 &1572.9&\textbf{2079.5}\\
Rowan & 410 & 446.5& 450.0 & 461.8 & 484.4 &505.0&&\\
Williamson&849 & 1136.9&1139.6&1154.4&1192.3&1215.7&1511.9&\\
\hline

\end{tabular}
}
\caption{(Top) The votes-by-round table for the 2022 election in the Bishopbriggs North and Campsie Ward of the East Dunbartonshire Council area. (Bottom) The votes-by-round table after adding 483 ballots with Hendry ranked first and Pews ranked last.}
\label{e-duns-table}

\end{table}

What happens if we add 482 (instead of 483) of these ballots to the original ballot data? Surprisingly, even though Pews wins over Williamson by about 61 votes when we add 483, when we add 482 Williamson edges out Pews by about 7 votes. To see why, note that in the bottom of Table \ref{e-duns-table} McDiarmid fails to achieve quota in the second round, but just barely. As a result, when McDiarmid does achieve quota in the third round, she has a healthy surplus to transfer to Pews. If we add only 482 ballots then the quota falls to 1711 and McDiarmid achieves quota in the second round by 0.1 votes. Transferring 0.1 votes has no meaningful electoral effect, and thus Williamson can go on to defeat Pews. It is funny that adding a single ballot can cause a swing of 68 votes in the penultimate round (going from $-7$ for Pews to $+61$).

We show the votes-by-round table when adding 483 ballots of the type above to the original ballot data, but we can actually add anywhere between 483 and 4,329 such ballots and Pews wins a seat over Williamson. Surprisingly, we can add \emph{thousands} of ballots on which Pews is ranked last (and the original election only contained 8,072 voters) and Pews wins a seat as a result. Furthermore, if we add these ballots one at a time, for a while Pews' margin of victory over Williamson only becomes more impressive. For example, if we add 2,000 of these ballots then, in the penultimate round of the resulting election, Pews defeats Williamson 1913.2 votes to 1530.1, a resounding victory of almost 400 votes (recall that in the original election Pews lost by $\approx$170 votes). 

\begin{figure}
\begin{center}

\includegraphics[scale=0.8]{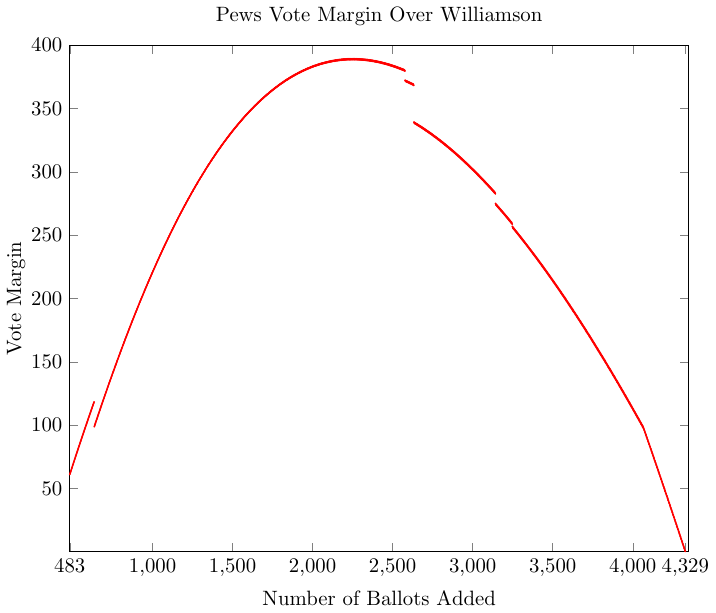} 

\end{center}
\caption{The amount of votes by which Pews defeats Williamson in the penultimate round of the election obtained by adding the given number of ballots with Pews ranked last to the original set of ballots.}
\label{e-duns-votes-margin}
\end{figure}

To see how the addition of these ballots strengthens Pews' win, Figure \ref{e-duns-votes-margin} provides a scatterplot which shows Pews' vote margin over Williamson. For each number of ballots added between 483 and 4329, the plot shows Pews' votes minus Williamson's in the penultimate round when Williamson is eliminated. Note that as we add ballots the vote margin curve can have jumps, which correspond to changes in how the election unfolds. For example, the first jump in the graph occurs when we move from 634 to 636 ballots added. When we add either amount of ballots both Ferretti and Hendry achieve quota in the first round, but when we add 634 Ferretti's vote total is higher than Hendry's and thus Ferretti's votes are transferred first (none of which are transferred to Hendry by rule), whereas for 636 the situation is reversed. Such discrete changes in how the election unfolds cause the ``jump discontinuities'' observed in the scatterplot.

\end{example}

In an STV election we generally do not care about obtaining a final ranking of the candidates. A candidate either wins a seat or does not, and whether a candidate wins the second versus fourth seat (or is eliminated early or later) does not matter. However, the method of STV can provide a ranking, where candidates are ranked by when they achieve quota (if multiple candidates are elected in the same round, they are ranked by vote total) and by when they are eliminated (an earlier elimination gives a lower ranking). For example, in the actual election for the above example STV ranks Ferretti first, McDiarmid second, Williamson third, and so on, until ranking Harris last. The next example provides the most dramatic rise in candidate rankings that we observed as a result of a paradox of negative involvement.

\begin{example} \textbf{East Lothian, 2022, Musselburgh Ward.}  The original election unfolds as shown in the top of Table \ref{e-lothian-table}. If we were to rank the candidates then Whyte comes in sixth as he is the second-to-last candidate to be eliminated. By contrast, McIntosh comes in second, earning the second seat in round 6.

\begin{table}
{\footnotesize
\begin{tabular}{c | c | c|c|c|c|c|c|c|c}
\multicolumn{10}{c}{\textbf{Actual Election}, Quota $=1472$}\\
\hline
Candidate & \multicolumn{9}{c}{Votes by Round}\\
\hline
Bennett&1074&1078.0&1081.0&1091.3&1160.9&1234.8&1276.5&1276.6&\textbf{1715.4}\\
Butts & 312 & 313.9 & 332.9&325.9&&&&\\
Carter&79&79.6&&&&&&&\\
Cassini&\textbf{1596}&&&&&&&&\\
Forrest&1209&1214.1&1222.2&1231.3&1287.4&1425.8&\textbf{1472.2}&&\\
Graham&84&85.7&92.8&&&&&&\\
Mackie&1082&1082.5&1100.5&1106.5&1161.5&1172.3&1177.2&1177.2&\\
McIntosh&991&1001.8&1017.8&1035.2&1100.7&\textbf{1666.9}&&&\\
Whyte&929&1022.2&1026.5&1056.4&1075.8&&&&\\
\end{tabular}

\vspace{.1 in}
\begin{tabular}{c | c | c|c|c|c|c|c|c|c}
\multicolumn{10}{c}{\textbf{Modified Election}, Quota $=1512$}\\
\hline
Candidate & \multicolumn{9}{c}{Votes by Round}\\
\hline
Bennett&1074&1082.2&1085.2&1095.7&1166.0&1323.9&1370.0&1372.7&\textbf{1813.6}\\
Butts & 312 & 315.8 & 325.0&328.0&&&&\\
Carter&79&80.3&&&&&&&\\
Cassini&\textbf{1796}&&&&&&&&\\
Forrest&1209&1219.4&1227.6&1236.8&1292.9&1516.0&\textbf{1516.0}&&\\
Graham&84&87.5&94.6&&&&&&\\
Mackie&1082&1083.0&1101.0&1107.0&1162.0&1192.1&1195.1&1195.4&\\
McIntosh&991&1044.6&1060.6&1078.4&1144.3&&&&\\
Whyte&929&1118.7&1123.2&1154.1&1173.9&\textbf{1659.0}&&&\\
\end{tabular}
\caption{(Top) The votes-by-round table for the 2022 election in the Musselburgh Ward of the East Lothian Council area. (Bottom) The votes-by-round table after adding 200 ballots with McIntosh ranked second and Whyte ranked last.}
\label{e-lothian-table}
}
\end{table}

Suppose we add 200 ballots of the form \begin{center}
Cassini$\succ$McIntosh$\succ$Forrest$\succ$Bennett$\succ$Butts$\succ$Carter$\succ$Graham$\succ$Mackie$\succ$Whyte
\end{center} to the ballot data. As with other examples, these ballots rank the original winners at the top of the ballots and so should only reinforce the original winner set. Furthermore, Whyte is ranked last and thus these ballots should not improve his electoral chances. However, as shown in the bottom of Table \ref{e-lothian-table}, the addition of these ballots causes Whyte to replace McIntosh in the winner set and in the modified election Whyte comes in second. By adding hundreds of ballots on which Whyte is ranked last and McIntosh is ranked second, we cause Whyte to jump from sixth place to second (and cause the corresponding fall for McIntosh). It is surprising that adding ballots which should be bad for Whyte cause him to achieve such a convincing electoral victory. That is, if we were first presented with the votes-by-round table in the bottom of Table \ref{e-lothian-table} then we may assume that Whyte definitely deserves a seat, yet Whyte only earns that seat after hundreds of his non-supporters cast ballots. We can add anywhere from 70 to 259 ballots of this form to engineer the paradox.

This example represents the most extreme type of a paradox of negative involvement as displayed in Example \ref{second-ex}. Adding more ballots of the type above allows Cassini to surpass quota by a larger amount, causing a much larger vote transfer to Whyte. We typically do not observe such a large jump in the rankings from this kind of paradox. 

\end{example}

The next example analyzes the only election we found in which multiple winners lose their seats as a result of a paradox of negative involvement.

\begin{example}\label{argyll-bute-ex}\textbf{Argyll and Bute, 2017, Isle of Bute Ward.} The votes-by-round table for the actual election is displayed in the top of Table \ref{argyll-bute-table}. Findlay and Moffat win seats in round 4 and Scoullar wins the third and final seat in round 7 after narrowly defeating Wallace by 33.2 votes.

\begin{table}

\begin{tabular}{c | c | c|c|c|c|c|c}
\multicolumn{8}{c}{\textbf{Actual Election}, Quota $=620$}\\
\hline
Candidate & \multicolumn{7}{c}{Votes by Round}\\
\hline
Findlay&433&437&461&\textbf{709}&&&\\
Gillies&325&343&&&&&\\
MacIntyre&395&402&433&&&&\\
McCallum&85&&&&&&\\
Moffat&472&494&588&\textbf{657}&&&\\
Scoullar&340&354&440&502&524.2&537.2&\textbf{740.5}\\
Wallace&427&437&482&492&496.8&504.0&\\
\end{tabular}

\vspace{.1 in}
\begin{tabular}{c | c | c|c|c|c|c|c}
\multicolumn{8}{c}{\textbf{Modified Election}, Quota $=623$}\\
\hline
Candidate & \multicolumn{7}{c}{Votes by Round}\\
\hline
Findlay&433&437&449&476&497.2&497.6&\\
Gillies&337&355&437&&&&\\
MacIntyre&395&402&444&491&519.0&519.5&\textbf{931.5}\\
McCallum&85&&&&&&\\
Moffat&472&494&614&\textbf{789}&&&\\
Scoullar&340&354&&&&&\\
Wallace&427&437&507&585&\textbf{629.0}&&\\
\end{tabular}\caption{(Top) The votes-by-round table for the 2017 election in the Isle of Bute Ward of the Argyll and Bute Council area. (Bottom) The votes-by-round table after adding 12 ballots with Findlay and Scoullar ranked in the top three and MacIntyre and Wallace ranked in the bottom two.}
\label{argyll-bute-table}

\end{table}

Suppose we add 12 ballots of the form \begin{center}
Gillies$\succ$Scoullar$\succ$Moffat$\succ$Findlay$\succ$Gillies$\succ$McCallum$\succ$MacIntyre$\succ$Wallace
\end{center} to the ballot data. Then the election would unfold as shown in the bottom of Table \ref{argyll-bute-table} and MacIntyre and Wallace replace Moffat and Scoullar in the winner set. This example  displays different dynamics than previous ones. In this case, we do not give a winning candidate more surplus votes which subsequently provide more transfer votes to a previously losing candidate. Instead, the addition of these 12 ballots allow Gillies to survive one round longer, causing Scoullar to be eliminated in the second round. Surprisingly, adding a handful of ballots on which Scoullar is ranked second causes him to fall from third place to sixth (out of seven). Eliminating Scoullar allows MacIntyre to survive to subsequent rounds, eventually winning a seat in the last round. The elimination of Scoullar also allows Moffat to achieve quota by a much larger margin than in the original election, allowing Wallace to earn the second seat in round 5. It is counterintuitive that adding a handful of ballots on which MacIntyre and Wallace are ranked at the bottom and Scoullar and Moffat are ranked in the top three could cause the former candidates to replace the latter in the winner set.

\end{example}

Finally, we present an example of an election which demonstrates paradoxes of negative and positive involvement. In this election we can add ballots such that the candidate ranked last on these ballots replaces the candidate ranked first on these ballots in the winner set.

\begin{example}\textbf{South Lanarkshire, 2022, Rutherglen Central and North Ward.} The original election unfolds as shown in the top of Table \ref{s-lanarks-table}. Lennon and Calikes win seats in rounds 2 and 3, respectively, and Cowan easily defeats Fox in round 5 to earn a seat in round 6. Note that Lennon achieves quota exactly in round 2, and thus does not transfer any surplus votes. If Lennon's election could be delayed so that he achieves quota by a sizable amount, McGinty would stand to benefit because McGinty is the only other candidate in the election who belongs to Lennons's party (Labour), and voters who rank Lennon first and rank another candidate tend to rank McGinty second.

Suppose we add four ballots of the form \begin{center}
Cowan$\succ$Calikes$\succ$Lennon$\succ$Adebo$\succ$Fox$\succ$McRae$\succ$McGinty
\end{center}
 to the ballot data. How could the addition of such ballots possibly be bad for Cowan? The bottom of Table \ref{s-lanarks-table} shows the effect of adding these ballots. With the additional ballots the quota is one vote higher, so that Lennon does not achieve quota in round 2. Instead Lennon achieves quota one round later and has 126 votes to transfer. Most of these votes go to McGinty, who now has enough support to win easily against Cowan in round 6. Because McGinty replaces Cowan in the winner set, this election demonstrates paradoxes of negative and positive involvement.

We can add up to 59 such ballots and achieve this double paradox.
\begin{table}

\begin{tabular}{c | c | c|c|c|c|c|c}
\multicolumn{8}{c}{\textbf{Actual Election}, Quota $=1270$}\\
\hline
Candidate & \multicolumn{7}{c}{Votes by Round}\\
\hline
Adebo&517&544&544&&&&\\
Calikes&1188&1265&1265&\textbf{1338}&&&\\
Cowan&725&796&796&840&899.6&1029.2&\textbf{1163.3}\\
Fox&609&612&612&730&730.4&881.5&\\
Lennon&1245&\textbf{1270}&&&&&\\
McGinty&531&567&567&716&717.9&&\\
McRae&261&&&&&&\\
\hline

\end{tabular}

\vspace{.1 in}
\begin{tabular}{c | c | c|c|c|c|c|c}
\multicolumn{8}{c}{\textbf{Modified Election}, Quota $=1271$}\\
\hline
Candidate & \multicolumn{7}{c}{Votes by Round}\\
\hline
Adebo&517&544&&&&&\\
Calikes&1188&1265&\textbf{1335}&&&&\\
Cowan&729&800&840&844.6&900.8&918.9&\\
Fox&609&612&721&725.4&725.8&&\\
Lennon&1245&1270&\textbf{1397}&&&&\\
McGinty&531&567&631&734.7&736.5&1044.0&\textbf{1356.7}\\
McRae&261&&&&&&\\
\hline

\end{tabular}
\caption{(Top) The votes-by-round table for the 2022 election in the Rutherglen Central and North Ward of the South Lanarkshire Council area. (Bottom) The votes-by-round table after adding 4 ballots with Cowan ranked first and McGinty ranked last.}
\label{s-lanarks-table}

\end{table}
\end{example}

We could find only one other election, the 2022 election in the Inverleith Ward of the City of Edinburgh Council area, which demonstrates paradoxes of negative and positive involvement. The dynamics are similar to the South Lanarkshire example:  adding extra ballots raises the quota, which delays the election of a candidate by one round so that this candidate has many more surplus votes to transfer.

\section{Conclusion}

Mathematically speaking, STV can be a funny voting method. As we have shown, the method can behave in pretty wild ways when we enlarge the electorate. More broadly, STV can also produce paradoxical outcomes when we remove ballots (producing no-show paradoxes) or when we shift candidates up or down the rankings of ballots (producing monotonicity paradoxes). Such strange outcomes point to a core of mathematical wackiness in the design of STV. Many of the examples in this article also illustrate STV's instability, in the sense that we can perturb the ballot data in relatively minor ways and produce large and paradoxical changes to the winner set in response. Whether the reader finds these issues problematic or not depends on their values and  approach to evaluating electoral procedures; at the very least, STV allows for some mathematical fun that other (perhaps more reasonable) methods do not.\\

\section{Appendix}

\subsection{How we searched the Scottish dataset for involvement paradoxes.} Given an election, we are unaware of necessary and sufficient conditions for that election to demonstrate an involvement paradox. Therefore, we cannot guarantee that the 99 elections listed below represent a complete list of elections which demonstrate a paradox in the Scottish dataset. Because we do not have conditions for checking for an involvement paradox, we have to be a bit creative in our search. We took the following approach, using Python 3 code. For each election, for each candidate $C$ we added bullet votes for $C$ (ballots with $C$ ranked first and no other candidates ranked) and checked for a change in the winner set. To be precise, if there were originally $k$ bullet votes for $C$ in the election then we added 1 bullet vote for $C$ to the original data and checked for a change in the winner set, then 2 bullet votes for $C$, etc., all the way up to $k$ additional bullet votes. We could not keep adding bullet votes for $C$ forever, and we decided to go up to doubling each candidate's bullet vote count and then move on to the next candidate.

We took this approach because for an involvement paradox to manifest we must find a set of identical ballots $\mathcal{B}$ to add to the original set of ballots, and some candidate must be ranked at the top of these ballots. For a given election, if our approach found a change in the winner set then we could attempt to fill in the bullet votes to a complete ranking and see if a paradox occurred.

For example, in the South Lanarkshire example (Example 6), the code reported that if we add 4 bullet votes for Cowan to the original ballot data then the winner set changes from $\{$Calikes, Cowan, Lennon$\}$ to $\{$Calikes, Lennon, McGinty$\}$. We then filled in those four bullet votes to ballots of the form 
\begin{center}
Cowan$\succ$Calikes$\succ$Lennon$\succ$Adebo$\succ$Fox$\succ$McRae$\succ$McGinty
\end{center}
and checked that McGinty would still replace Cowan in the winner set. Once we knew a paradox was present, we used additional code to find the range of ballots which can be added, which is how we can report a maximum of 59 such ballots to engineer the paradox.

This methodology using bullet votes can only provide a list of elections to investigate, it cannot guarantee that a paradox will be found. There were 31 elections in which adding bullet votes for at least one candidate caused a change to the winner set, but we could not find a way to engineer an involvement paradox. It is possible that we missed some elections which demonstrate an involvement paradox; the bullet vote methodology was meant to cast a wide net, but it is possible the net was not wide enough.

\subsection{Elections which demonstrate an involvement paradox.} The elections are reported in the format Council Area, Year, Ward. We also report party information for the winning candidates, using the following abbreviations: Conservative (Con), Green (Grn), Independent (Ind), Labour (Lab), Liberal Democrat (LD), Scottish National Party (SNP). For each election we provide the type of ballot we add in order to engineer an involvement paradox; all such reported ballots provide a complete ranking of the candidates in the given election.

\begin{itemize}
\item \textbf{Aberdeenshire, 2012, Ward 7 (Turriff and District)}


Winners: Duncan (SNP), Norrie (Ind), Robertson (LD)

If we add 40-70 ballots of the form below then Mair (SNP) replaces Norrie in the winner set.

\begin{center}
Duff$\succ$Norrie$\succ$Duncan$\succ$Robertson$\succ$Ogston$\succ$Mair
\end{center}

\item  \textbf{Aberdeenshire, 2017, Ward 5 (Peterhead North and Rattray)}

Winners: Allan (SNP), Beagrie (Con), Buchan (Ind), Sutherland (Ind)

If we add 97-526 ballots of the form below then McRae (SNP) replaces Sutherland in the winner set.

\begin{center}
Allan$\succ$Sutherland$\succ$Buchan$\succ$Beagrie$\succ$Massey$\succ$McRae
\end{center}

\item \textbf{Aberdeenshire, 2017, Ward 9 (Ellon and District)} (downward paradox election)

Winners: Davidson (LD), Kahanov-Kloppert (SNP), Owen (Con), Thomson (SNP)

If we add 21-466 ballots of the form below then Morgan (Lab) replaces Kahanov-Kloppert in the winner set.

\begin{center}
Davidson$\succ$Kahanov-Kloppert$\succ$Owen$\succ$Thomson$\succ$Morgan\\
\end{center}

\item \textbf{Aberdeenshire, 2022, Ward 4 (Central Buchan)}

Winners: Crowson (SNP), Mair (SNP), Powell (Con), Simpson (LD)

If we add 26-233 ballots of the form below then Owen (Con) replaces Crowson in the winner set.

If we add 234-2856 ballots of the form below then Owen replaces Mair in the winner set.

\begin{center}
Powell$\succ$Crowson$\succ$Simpson$\succ$Mair$\succ$Cole-Hamilton

$\succ$Cross$\succ$MahMood$\succ$Moore$\succ$Smith$\succ$Owen
\end{center}

\item \textbf{Aberdeenshire, 2022, Ward 8 (Mid-Formartine)}


Winners: Hassan (LD), Johnston (Ind), Nicol (SNP), Ritchie (Con)

If we add 25-94 ballots of the form below then Powell (Con) replaces Hassan in the winner set.

\begin{center}
Hutchison$\succ$Hassan$\succ$Nicol$\succ$Johnston$\succ$Ritchie$\succ$Powell
\end{center}

\item \textbf{Aberdeenshire, 2022, Ward 19 (Mearns)}

Winners: Carnie (Con), Carr (Con), Evison (Ind), Stelfox (SNP)

If we add 163-395 ballots of the form below then Ewen (LD) replaces Carnie in the winner set.

\begin{center}
Stelfox$\succ$Evison$\succ$Carnie$\succ$Carr$\succ$Allan$\succ$Fraser

$\succ$Laurenson$\succ$Neill$\succ$Stewart$\succ$Wilson$\succ$Ewen
\end{center}

If we add 396-510 ballots of the form below Stewart (Ind) replaces Carnie in the winner set.

\begin{center}
Stelfox$\succ$Evison$\succ$Carnie$\succ$Carr$\succ$Allan$\succ$Fraser

$\succ$Laurenson$\succ$Neill$\succ$Wilson$\succ$Ewen$\succ$Stewart
\end{center}

\item \textbf{Aberdeen City, 2012, Ward 2 (Torry - Ferryhill)}

Winners: Allan (Lab), Kiddie (SNP), Dickson (SNP), Donnelly (Con)

If we add 399-540 ballots of the form below then Russell (Lab) replaces Dickson in the winner set.

If we add 541-3350 ballots  of the form below then Russell replaces Donnelly in the winner set.

\begin{center}
Allan$\succ$Kiddie$\succ$Dickson$\succ$Donnelly$\succ$Brown$\succ$Fryer

$\succ$Kelly$\succ$MacKay$\succ$Reekie$\succ$Robertson$\succ$Watson$\succ$Russell
\end{center}

\item \textbf{Aberdeen City, 2017, Ward 1 (Dyce/Bucksburn/Danestone)}

Winners: Crockett (Lab), MacGregor (SNP), Samari (SNP), MacKenzie (Con)

If we add 832-1439 ballots of the form below then Duthie (Ind) replaces Samari in the winner set.

\begin{center}
MacKenzie$\succ$Crockett$\succ$Samari$\succ$MacGregor$\succ$Lawrence$\succ$Pearce$\succ$Duthie
\end{center}

If we add 1440-2585 ballots of the form below thenPearce (LD) replaces Samari in the winner set.

\begin{center}
MacKenzie$\succ$Crockett$\succ$Samari$\succ$MacGregor$\succ$Lawrence$\succ$Duthie$\succ$Pearce
\end{center}

\item \textbf{Aberdeen City, 2017, Ward 2 (Bridge of Don)}

Winners: Alphonse (SNP), Hunt (Con), Reynolds (Ind), Stuart (SNP)

If we add 242-2909 ballots of the form below then Farquhar (LD) replaces Stuart in the winner set.

\begin{center}
Hunt$\succ$Reynolds$\succ$Stuart$\succ$Alphonse$\succ$Irving-Lewis

$\succ$McLean$\succ$Saunders$\succ$Young$\succ$Farquhar
\end{center}

\item \textbf{Aberdeen City, 2022, Ward 8 (George St-Harbour)}

Winners: Bouse (LD), Henrickson (SNP), Hutchison (SNP), MacDonald (Lab)

If we add 25-149 Ballots of the form below then Ingerson (Grn) replaces Bouse in the winner set.

\begin{center}
Painter$\succ$Bouse$\succ$Henrickson$\succ$Hutchison$\succ$MacDonald$\succ$Chaudry$\succ$Ingerson
\end{center}

\item \textbf{Angus, 2017, Ward 4 (Monifieth and Sidlaw)}

Winners: Fotheringham (Con), Hands (SNP), Lawrie (LD), Whiteside (SNP)

If we add 291-1907 ballots of the form below then Strachan (Lab) replaces Whiteside in the winner set.

\begin{center}
Fotheringham$\succ$Lawrie$\succ$Whiteside$\succ$Hands$\succ$Strachan
\end{center}

\item \textbf{Angus, 2022, Ward 6 (Arbroath West Letham and Friockheim)}


Winners: Cowdy (SNP), Fairweather (Ind), Nicol (Con), Shepherd (SNP)

If we add 75-385 ballots of the form below then Wren (Ind) replaces Fairweather in the winner set.

\begin{center}
Ruddy$\succ$Fairweather$\succ$Cowdy$\succ$Nicol$\succ$Shepherd

$\succ$Campbell$\succ$Falconer$\succ$Keogh$\succ$Vivers$\succ$Wren
\end{center}

\item \textbf{Argyll and Bute, 2012, Ward 5 (Oban North and Lorn)}

Winners: Glen-Lee (SNP), MacDonald (Ind), MacIntyre (Ind), Robertson (Ind)

If we add 14-479 ballots of the form below then Melville (SNP) replaces MacDonald in the winner set.

\begin{center}
Glen-Lee$\succ$MacDonald$\succ$MacIntyre$\succ$Robertson$\succ$Doyle
$\succ$MacKay$\succ$McIntosh$\succ$Neal$\succ$Rutherford$\succ$Melville
\end{center}

\item \textbf{Argyll and Bute, 2017, Ward 8 (Isle of Bute)}


Winners: Findlay (SNP), Moffat (Ind), Scoullar (Ind)

If we add 12-18 ballots of the form below then MacIntyre (SNP) and Wallace (Con) replace Moffat and Scoullar in the winner set. (We could also rank MacIntyre last and Wallace second-to-last.)

If we add 20-39 ballots of the form below then Wallace replaces Scoullar in the winner set.

\begin{center}
Gillies$\succ$Scoullar$\succ$Moffat$\succ$Findlay$\succ$Gillies$\succ$McCallum$\succ$MacIntyre$\succ$Wallace
\end{center}

\item \textbf{Argyll and Bute, 2021 by-election, Ward 8 (Isle of Bute)}

Winner: McCabe (Ind)

If we add 26-38 ballots of the form below then Findlay (SNP) becomes the winner.

\begin{center}
Gillies$\succ$McCabe$\succ$MacDonald$\succ$Wallace$\succ$Findlay
\end{center}

\item \textbf{Argyll and Bute, 2022, Ward 4 (Oban South and the Isles)}

Winners: Hampsey (Con), Hume (SNP), Kain (Ind), Lynch (SNP)

If we add 4-8 ballots of the form below then Meyer (Grn) replaces Kain in the winner set.

\begin{center}
Lynch$\succ$Hume$\succ$Kain$\succ$Hampsey$\succ$Ageer$\succ$Boswell

$\succ$Campbell$\succ$Kennedy$\succ$McGrigor$\succ$Watson$\succ$Meyer
\end{center}

\item \textbf{Argyll and Bute, 2022, Ward 8 (Isle of Bute)}

Winners: Kennedy-Boyle (SNP), McCabe (Ind), Wallace (Con)

If we add 18-140 ballots of the form below then Moffat (Ind) replaces Wallace in the winner set.

\begin{center}
Kennedy-Boyle$\succ$Wallace$\succ$McCabe$\succ$Gillies$\succ$Malcolm

$\succ$McFarlane$\succ$McGowan$\succ$Stuart$\succ$Moffat
\end{center}

\item \textbf{Argyll and Bute, 2022, Ward 9 (Lomond North)}

Winners: Corry (Con), Irvine (Ind), Paterson (SNP)

If we add 63-1003 ballots of the form below then Freeman (Ind) replaces Irvine in the winner set.

\begin{center}
Corry$\succ$Paterson$\succ$Irvine$\succ$MacIntyre$\succ$Millar$\succ$Robinson$\succ$Freeman
\end{center}

\item \textbf{City of Edinburgh, 2012, Ward 4 (Forth)}

Winners: Cardownie (SNP), Day (Lab), Jackson (Con), Redpath (Lab)

If we add 1002-3052 ballots of the form below then Gordon (SNP) replaces Redpath in the winner set.

\begin{center}
Jackson$\succ$Day$\succ$Cardownie$\succ$Redpath$\succ$Henderson

$\succ$Joester$\succ$MacMhicean$\succ$Wight$\succ$Gordon
\end{center}

\item \textbf{City of Edinburgh, 2012, Ward 16 (Liberton Gilmerton)}

Winners: Austin-Hart (Lab), Buchanan (SNP), B. Cook (Lab), N. Cook (Con)

If we add 36-1480 ballots of the form below then Howie (SNP) replaces N. Cook in the winner set.
\begin{center}
Buchanan$\succ$N. Cook$\succ$B. Cook$\succ$Austin-Hart$\succ$Carter$\succ$Fox$\succ$Knox$\succ$Howie\\
\end{center}

\item \textbf{City of Edinburgh, 2017, Ward 4 (Forth)}

Winners: Bird (SNP), Campbell (Con), Day (Lab), Gordon (SNP)

If we add 931-3079 ballots of the form below then Wight (LD) replaces Gordon in the winner set. 

\begin{center}
Campbell$\succ$Day$\succ$Gordon$\succ$Bird$\succ$MacKay$\succ$Pugh$\succ$Ross$\succ$Wight \\
\end{center}

\item \textbf{City of Edinburgh, 2017, Ward 5 (Inverleith)}

Winners: Barrie (SNP), Mitchell (Con), Osler (LD), Whyte (Con)

If we add 508-4404 ballots of the form below then Bagshaw (Grn) replaces Whyte in the winner set. 

If we add 4405-10935 ballots of the form below then Bagshaw replaces Mitchell in the winner set.

\begin{center}
Barrie$\succ$Osler$\succ$Whyte$\succ$Mitchell$\succ$Woolnough$\succ$Dalgleish$\succ$Laird$\succ$Bagshaw\\
\end{center}

\item  \textbf{City of Edinburgh, 2017, Ward 7 (Sighthill/Gorgie)}

Winners: Dixon (SNP), Fullerton (SNP), Graczyk (Con), Wilson (Lab)

If we add 678-1647 ballots of the form below then Heap (Grn) replaces Dixon in the winner set.

\begin{center}
Graczyk$\succ$Wilson$\succ$Dixon$\succ$Fullerton$\succ$Hayter$\succ$Scobie$\succ$Smith$\succ$Strange$\succ$Heap
\end{center}

\item \textbf{City of Edinburgh, 2017, Ward 8 (Colinton/Fairmilehead)}

Winners: Arthur (Lab), Doggart (Con), Rust (Con)

If we add 101-1605 ballots of the form below then Lewis (SNP) replaces Doggart in the winner set.

\begin{center}
Arthur$\succ$Doggart$\succ$Rust$\succ$Walker$\succ$Marsden$\succ$Lewis
\end{center}

\item  \textbf{City of Edinburgh, 2022, Ward 2 (Pentland Hills)}

Winners: Bruce (Con), Gardiner (SNP), Glasgow (SNP), Jenkinson (Lab)

If we add 24-2060 ballots of the form below then Gilbhrist (Con) replaces Glasgow in the winner set.

\begin{center}
Bruce$\succ$Glasgow$\succ$Jenkinson$\succ$Gardiner$\succ$Chappell$\succ$Fettes$\succ$Muller$\succ$Rowlands$\succ$Gilchrist
\end{center}

\item \textbf{City of Edinburgh, 2022, Ward 5 (Inverleith)}


Winners: Bandel (Grn), Mitchell (Con), Nicolson (SNP), Osler (LD)

If we add 71-113 ballots of the form below then Munro-Brian (Lab) replaces Bandel in the winner set.

\begin{center}
Bandel$\succ$Mitchell$\succ$Nicolson$\succ$Osler$\succ$Wood$\succ$

Herring$\succ$Holden$\succ$Laird$\succ$McNamara$\succ$Munro-Brian
\end{center}

\item \textbf{City of Edinburgh, 2022, Ward 11 (City Center)}

Winners: Graham (Lab), McFarlane (SNP), Miller (Grn), Mowat (Con)

If we add 71-2103 ballots of the form below then Foxall (LD) replaces Graham in the winner set.

\begin{center}
Mowat$\succ$Miller$\succ$Graham$\succ$McFarlane$\succ$Bob$\succ$Carson$\succ$Illingworth

$\succ$Mwiki$\succ$Pakpahan-Campbell$\succ$Penman$\succ$Rowan$\succ$Shaw$\succ$Foxall
\end{center}

\item \textbf{City of Edinburgh, 2022, Ward 16 (Liberton-Gilmerton)}

Winners: Cameron (Lab), Doggart (Con), MacInnes (SNP), Mattos-Coelho (SNP)

If we add 305-6574 ballots of the form below then Measom (Lab) replaces Mattos-Coelho in the winner set.

\begin{center}
Cameron$\succ$Doggart$\succ$Mattos-Coelho$\succ$MacInnes

$\succ$Christie$\succ$Fox$\succ$Meron$\succ$Nichol$\succ$Planche$\succ$Measom
\end{center}

We can also make the paradox occur in a weaker fashion so that Nichol (Grn) earns a seat. If we add 672-898 ballots of the form below then Nichol replaces Mattos-Coelho in the winner set.

\begin{center}
Doggart$\succ$MacInnes$\succ$Cameron$\succ$Christie$\succ$Fox$\succ$Measom

$\succ$Mattos-Coelho$\succ$Meron$\succ$Planche$\succ$Nichol
\end{center}

\item \textbf{Clackmannanshire, 2017, Ward 3 (Clackmannanshire Central)}

Winners: Fairlie (SNP), Stewart (Lab), Watson (Con)

If we add 279-1004 ballots of the form below then Watt replaces Fairlie in the winner set.

\begin{center}
Stewart$\succ$Fairlie$\succ$Watson$\succ$Short$\succ$Wilkinson$\succ$Watt
\end{center}

\item \textbf{Dumfries and Galloway, 2017, Ward 9 (Nith)}


Winners: Campbell (SNP), Johnstone (Con), Martin (Lab), Murray (Lab)

If we add 317-1161 ballots of the form below then Slater (Ind) replaces Martin in the winner set.

\begin{center}
Johnstone$\succ$Campbell$\succ$Martin$\succ$Murray$\succ$Cowan$\succ$

Crosbie$\succ$Dennis$\succ$Rogerson$\succ$Witts$\succ$Slater
\end{center}

\item \textbf{Dumfries and Galloway, 2022, Ward 12 (Annandale East and Eskdale)}


Winners: Carruthers (Con), Dryburgh (Lab), Male (Ind)

If we add 27-93 ballots of the form below then Tait (Con) replaces Male in the winner set.

\begin{center}
Willmot$\succ$Male$\succ$Carruthers$\succ$Dryburgh$\succ$Herbst-Gray$\succ$Mattock$\succ$Tait
\end{center}

\item \textbf{Dundee, 2012, Ward 8 (Ferry)}

Winners: Bidwell (Lab), Cordell (SNP), Scott (Con)

If we add 185-285 ballots of the form below then Wallace (Con) replaces Cordell in the winner set.

If we add 186-2673 ballots of the form below then Wallace replaces Bidwell in the winner set.

\begin{center}
Scott$\succ$Cordell$\succ$Bidwell$\succ$Duncan$\succ$Guild$\succ$Wallace
\end{center}

\item \textbf{East Ayrshire, 2012, Ward 5 (Kilmarnock South)}

Winners: Knapp (Lab), Ross (SNP), Todd (SNP)

If we add 11-754 ballots of the form below then Scott (Lab) replaces Ross in the winner set.

\begin{center}
Knapp$\succ$Todd$\succ$Ross$\succ$Holden$\succ$Scott
\end{center}

\item \textbf{East Ayrshire, 2012, Ward 7 (Ballochmyle)}

Winners: McGhee (Lab), Primrose (SNP), Roberts (SBP), Shaw (Lab)

If we add 67-334 ballots of the form below then Murray (Lab) replaces Primrose in the winner set.

\begin{center}
McGhee$\succ$Primrose$\succ$Roberts$\succ$Shaw$\succ$Martin$\succ$Murray
\end{center}

\item \textbf{East Ayrshire, 2012, Ward 9 (Doon Valley)}

Winners: Bell (SNP), Dinwoodie (Lab), Pirie (Lab)

If we add 129-311 ballots of the form below then Borthwick (Ind) replace Pirie in the winner set.

\begin{center}
Bell$\succ$Pirie$\succ$Dinwoodie$\succ$Filson$\succ$Grant$\succ$Borthwick
\end{center}

\item \textbf{East Ayrshire, 2017, Ward 6 (Irvine Valley)}

Winners: Cogley (Rubbish), Mair (Lab), Whitham (SNP)

If we add 33-34 ballots of the form below then McFadzean (Con) replaces Mair in the winner set.

\begin{center}
Whitham$\succ$Mair$\succ$Cogley$\succ$Brannagan$\succ$Gartland$\succ$King$\succ$Young$\succ$McFadzean
\end{center}

\item \textbf{East Ayrshire, 2017, Ward 8 (Cumnock and New Cumnock)}

Winners: Crawford (Lab), McMahon (SNP), Todd (SNP), Young (Con)

If we add 22-3556 ballots of the form below then Mochan (Lab) replaces McMahon in the winner set.

\begin{center}
Crawford$\succ$Young$\succ$Todd$\succ$McMahon$\succ$Bircham$\succ$Owens$\succ$Black$\succ$Mochan
\end{center}

\item \textbf{East Dunbartonshire, 2017, Ward 6 (Lenzie and Kirkintilloch South)}


Winners: Ackland (LD), Renwick (SNP), Thornton (Con)

If we add 123-174 ballots of the form below then Taylor (Ind) replaces Ackland in the winner set.

\begin{center}
Geekie$\succ$Ackland$\succ$Renwick$\succ$Thornton$\succ$Robertson$\succ$Scrimgeour$\succ$Sinclair$\succ$Taylor
\end{center}

\item \textbf{East Dunbartonshire, 2022, Ward 4 (Bishopbriggs North \& Campsie)}

Winners: Ferretti (SNP), Hendry (Con), McDiarmid (Lab), Williamson (SNP)

If we add 483-4329 ballots of the form below then Pews (LD) replaces Williamson in the winner set.

\begin{center}
Hendry$\succ$McDiarmid$\succ$Williamson$\succ$Ferretti$\succ$Gallacher$\succ$Harris$\succ$Rowan$\succ$Pews
\end{center}

\item \textbf{East Lothian, 2017, Ward 1 (Musselburgh)}

Winners: Currie (SNP), Forrest (Lab), Mackie (Con), Williamson (SNP)

If we add 793-806 ballots of the form below then Caldwell (Ind) replaces Williamson in the winner set.

\begin{center}
Mackie$\succ$Currie$\succ$Forrest$\succ$Williamson$\succ$Graham

$\succ$Rose$\succ$Sangster$\succ$Sives$\succ$Caldwell\\
\end{center}

\item \textbf{East Lothian, 2022, Ward 1 (Musselburgh)}

Winners: Bennett (Lab), Cassini (SNP), Forrest (SNP), McIntosh (Grn)

If we add 70-259 ballots of the form below then Whyte replaces McIntosh in the winner set.

\begin{center}
Cassini$\succ$McIntosh$\succ$Forrest$\succ$Bennett$\succ$Butts$\succ$Carter$\succ$Graham$\succ$Mackie$\succ$Whyte
\end{center}

\item \textbf{Eilean Siar, 2012, Ward 7 (Steornabhagh a Tuath)}

Note: This is a weak example of the paradox. We cannot put all of the winners at the top of the ballot, and the candidate who gets replaced cannot be put in the top $S$ rankings.

Winners: MacAulay (Ind), R. MacKay (Ind), MacKenzie (Ind), G. Murray (SNP)

If we add 4-72 ballots of the form below then Ahmed (SNP) replaces MacAulay in the winner set.

\begin{center}
Campbell$\succ$R. MacKay$\succ$MacKenzie$\succ$G. Murray$\succ$MacAulay

$\succ$J. MacKay$\succ$Morrison$\succ$M. Murray$\succ$Paterson$\succ$Ahmed
\end{center}

\item \textbf{Falkirk, 2012, Ward 1 (Bo'ness and Blackness)}

Winners: Mahoney (Lab), Ritchie (SNP), Turner (SNP)

If we add 12-1378 ballots of the form below then Aitchison (Lab) replaces Turner in the winner set.

\begin{center}
Mahoney$\succ$Turner$\succ$Ritchie$\succ$Munro$\succ$Aitchison
\end{center}

\item \textbf{Falkirk, 2012, Ward 9 (Upper Braes)}

Winners: Hughes (SNP), McLuckie (Lab), Murray (Lab)

If we add 115-818 ballots of the form below then Wilson (SNP) replaces Murray in the winner set.

\begin{center}
Hughes$\succ$Murray$\succ$McLuckie$\succ$Rust$\succ$Wilson
\end{center}

\item \textbf{Falkirk, 2022, Ward 2 (Grangemouth)}


Winners: Balfour (SNP), Nimmo (Lab), Spears (Ind)

If we add 9- ballots of the form below then Haston (SNP) replaces Spears in the winner set.

\begin{center}
Bryson$\succ$Spears$\succ$Balfour$\succ$Nimmo$\succ$Bozza$\succ$Martin$\succ$Stenhouse$\succ$Haston
\end{center}

\item \textbf{Falkirk, 2022, Ward 6 (Falkirk North)}

Winners: Bissett (Lab), Bundy (Con), Meiklejohn (SNP), Sinclair (SNP)

If we add 86-164 ballots of the form below then Burgess (Lab) replaces Bundy in the winner set.

\begin{center}
Bissett$\succ$Bundy$\succ$Meiklejohn$\succ$Sinclair$\succ$Arshad$\succ$McLaughlin$\succ$Burgess
\end{center}

\item \textbf{Fife, 2012, Ward 6 (Inverkeithing and Dalgety Bay)}

Winners: Dempsey (Con), Laird (Lab), McGarry (SNP), Yates (Lab)

If we add 61-1133 ballots of the form below then Todd (SNP) replaces Yates in the winner set.

If we add 1134-4126 ballots of the form below then Todd replaces Laird in the winner set.

\begin{center}
McGarry$\succ$Dempsey$\succ$Yates$\succ$Laird$\succ$Arthur$\succ$Walker$\succ$Todd
\end{center}

\item \textbf{Fife, 2017, Ward 6 (Inverkeithing and Dalgety Bay)}

Winners: Barratt (SNP), Dempsey (Con), Laird (Lab), McGarry (SNP)

If we add 310-583 ballots of the form below then Cannon-Todd (Ind) replaces Barratt in the winner set.

\begin{center}
Laird$\succ$Barratt$\succ$McGarry$\succ$Dempsey$\succ$Hansen$\succ$Hawthorne$\succ$Cannon-Todd
\end{center}

\item \textbf{Fife, 2017, Ward 12 (Kirkaldy East)}


Winners: Cameron (Lab), Cavanagh (SNP), Watt (Con)

If we add 28-365 ballots of the form below then Penman (Ind) replaces Watt in the winner set.

\begin{center}
Cameron$\succ$Watt$\succ$Cavanagh$\succ$Cameron$\succ$Forbes

$\succ$McMahon$\succ$Ritchie$\succ$Rottger$\succ$Penman\\
\end{center}

\item \textbf{Fife, 2022, Ward 5 (Rosyth)}

Winners: Verrecchia (Lab), Goodall (SNP), Jackson (SNP)

If we add 136-540 ballots of the form below then Thomson (Con) Jackson in the winner set.

\begin{center}
Verrecchia$\succ$Goodall$\succ$Jackson$\succ$Kittle$\succ$Lynas

$\succ$McIntyre$\succ$McOwan$\succ$Morton$\succ$Thomson
\end{center}

\item \textbf{Fife, 2022, Ward 7 (Cowdenbeath)}

Winners: Campbell (Lab), Bain (SNP), Watt (Con), Robb (SNP)

If we add 9-3001 ballots of the form below then Guichan (Lab) replaces Robb in the winner set.

\begin{center}
Campbell$\succ$Bain$\succ$Watt$\succ$Robb$\succ$Bijster$\succ$Greig$\succ$Venturi$\succ$Lee$\succ$Guichan
\end{center}

\item \textbf{Fife, 2022, Ward 12 (Kirkaldy East)}

Winners: Cameron (Lab), Cavanagh (SNP), Patrick (SNP)

If we add 89-200 ballots of the form below then MacDonald (Lab) replaces Patrick in the winner set.

\begin{center}
Cameron$\succ$Patrick$\succ$Cavanagh$\succ$Hansen$\succ$Neilson$\succ$Thompson$\succ$Watt$\succ$MacDonald
\end{center}

\item \textbf{Glasgow City, 2007, Baillieston}

Winners: Coleman (Lab), Hay (Lab), Mason (SNP), McDonald (SNP)

If we add 596-5311 ballots of the form below then MacBean (Lab) replaces McDonald in the winner set.

\begin{center}
Colman$\succ$Mason$\succ$Hay$\succ$McDonald$\succ$Clark$\succ$Dickie

$\succ$Kayes$\succ$McVicar$\succ$Morrison$\succ$Watt$\succ$MacBean
\end{center}

\item \textbf{Glasgow City, 2007, Craigton}

Winners: Black (Sol), Gibson (SNP), Kerr (Lab), Watson (Lab)

If we add 491-8946 ballots of the form below then MacDiarmid (Lab) replaces Black in the winner set.

\begin{center}
Kerr$\succ$Gibson$\succ$Watson$\succ$Black$\succ$Coghill$\succ$Dingwall
$\succ$Masterton$\succ$McGartland$\succ$Petty$\succ$MacDiarmid
\end{center}

\item \textbf{Glasgow City, 2007, Ward 6 (Pollokshields)}

Winners: Malik (SNP), Meikle (Con), Rabbani (Lab).

If we add 378-644 ballots of the form below then Ruffell (Grn) replaces Meikle in the winner set.

\begin{center}
Malik$\succ$Meikle$\succ$Rabbani$\succ$Ashraf$\succ$Currie$\succ$Nelson$\succ$Shoaib$\succ$Uygun$\succ$Ruffell
\end{center}

\item \textbf{Glasgow City, 2017, Ward 9 (Calton)}

Winners: Connelly (Con), Hepburn (SNP), Layden (SNP), O'Lone (Lab)

If we add 123-952 ballots of the form below then Rannachan (Lab) replaces Connelly in the winner set.

\begin{center}
O'Lone$\succ$Connelly$\succ$Hepburn$\succ$Layden$\succ$MacPherson

$\succ$McGurk$\succ$McLaren$\succ$Pike$\succ$Rannachan
\end{center}

\item \textbf{Glasgow City, 2017, Ward 19 (Shettleston)}

Winners: Doherty (SNP), Ferns (SNP), T. Kerr (Con), McAveety (Lab)

If we add 521-7491 ballots of the form below then Simpson (Lab) replaces Ferns in the winner set.

\begin{center}
McAveety$\succ$Doherty$\succ$Ferns$\succ$T. Kerr$\succ$Campbell$\succ$Cocozza

$\succ$Corran$\succ$A. Kerr$\succ$Marshall$\succ$Pollard$\succ$Robertson$\succ$Simpson
\end{center}

\item \textbf{Glasgow City, 2022, Ward 7 (Langside)}

Winners: Aitken (SNP), Bruce (Grn), Docherty (Labour), Leinster (SNP)

If we add 83-3479 ballots of the form below then McKenzie (Lab) replaces Leinster in the winner set.

\begin{center}
Docherty$\succ$Leinster$\succ$Aitken$\succ$Bruce$\succ$Osuchukwu

$\succ$Shields$\succ$Stevenson$\succ$Whyte$\succ$McKenzie
\end{center}

\item \textbf{Glasgow City, 2022, Ward 13 (Garscadden-Scotstounhill)}


Winners: Butler (Lab), Cunningham (SNP), Mitchell (SNP), Murray (Lab)

If we add 571-678 ballots of the form below then Ugbah (SNP) replaces Cunningham in the winner set.

If we add 679-1217 ballots of the form below then Morrison and Ugbah replace Mitchell and Murray in the winner set.

\begin{center}
Morrison$\succ$Mitchell$\succ$Butler$\succ$Cunningham$\succ$Murray$\succ$Hamelink$\succ$Waterfield$\succ$Ugbah
\end{center}

\item \textbf{Glasgow City, 2022, Ward 19 (Shettleston)}

Winners: Doherty (SNP), Kerr (Con), McAveety (Lab), Pidgeon (Lab)

If we add 100-3212 ballots of the form below then Turner replaces Kerr in the winner set.

\begin{center}
Doherty$\succ$Pidgeon$\succ$Kerr$\succ$McAveety$\succ$Christie$\succ$McLaughlin$\succ$Sullivan$\succ$Turner
\end{center}

\item \textbf{Glasgow City, 2022, Ward 23 (Patrick East-Kelvindale)}


Winners: Anderson (Grn), Brown (Lab), Johnstone (Lab), McLean (SNP)

If we add 123-558 ballots of the form below then Asghar (Con) replaces Johnstone in the winner set.

\begin{center}
Wilson$\succ$Johnstone$\succ$Brown$\succ$Anderson$\succ$McLean

$\succ$McMillan$\succ$Moohan$\succ$Nwaokorobia$\succ$Asghar
\end{center}

\item \textbf{Highland, 2012, Ward 7 (Cromarty Firth)}

Winners: Finlayson (Ind), Rattray (LD), Smith (SNP), Wilson (Ind)

If we add 1-58 ballots of the form below then Fletcher (SNP) replaces Rattray in the winner set.

\begin{center}
Smith$\succ$Rattray$\succ$Wilson$\succ$Finlayson$\succ$MacInnes$\succ$McCaffery$\succ$Rous$\succ$Fletcher
\end{center}

\item \textbf{Highland, 2012, Ward 20 (Inverness South)}

Winners: Caddick (LD), Crawford (Ind), Gowans (SNP), Prag (LD)

If 23 add 15-617 ballots of the form below then B. Boyd (SNP) replaces Prag in the winner set.

\begin{center}
Gowans$\succ$Prag$\succ$Caddick$\succ$Crawford$\succ$Bonsor

$\succ$D. Boyd$\succ$MacKenzie$\succ$B. Boyd
\end{center}

If we add 83-133 ballots of the form below then MacKenzie (Lab) replaces Prag in the winner set.

\begin{center}
Crawford$\succ$Caddick$\succ$Prag$\succ$Gowans$\succ$Bonsor

$\succ$B. Boyd$\succ$D. Boyd$\succ$MacKenzie
\end{center}

\item \textbf{Highland, 2012, Ward 22 (Fort William and Ardnamurchan)}

Winners: Baxter (Ind), Gormley (SNP), MacLennan (Ind), Murphy (Lab).

If we add 21-146 ballots of the form below  then Corrigan (Ind) replaces MacLennan in the winner set.

\begin{center}
Baxter$\succ$MacLennan$\succ$Gormley$\succ$Murphy$\succ$Gillespie

$\succ$MacDonald$\succ$Mackie$\succ$Corrigan
\end{center}

\item \textbf{Highland, 2017, Ward 2 (Thurso and Northwest Caithness)}

Winners: MacKay (Ind), Mackie (Con), Reiss (Ind), Rosie (SNP)

If we add 259-848 ballots of the form below then Coghill (Ind) replaces MacKay in the winner set.

\begin{center}
Reiss$\succ$Rosie$\succ$MacKay$\succ$Mackie$\succ$Saxon$\succ$Farmer$\succ$Glasgow$\succ$Owsnett$\succ$Coghill
\end{center}

\item \textbf{Highland, 2017, Ward 4 (East Sutherland and Edderton)}


Winners: Gale (LD), MacKay (Lab), McGillivray (Ind)

If we add 117-184 ballots of the form below then Phillips (SNP) replaces McGillivray in the winner set.

\begin{center}
Short$\succ$McGillivray$\succ$Gale$\succ$MacKay$\succ$Christian$\succ$Gunn$\succ$Phillips
\end{center}

\item \textbf{Highland, 2017, Ward 17 (Culloden and Ardersier)}


Winners: Balfour (Ind), Campbell-Sinclair (SNP), Robertson (LD)

If we add 25-36 ballots of the form below then Cullen (Con) replaces Robertson in the winner set.

If we add 37-488 ballots of the form below then Cullen and Munro replace Campbell-Sinclair and Robertson in the winner set.

\begin{center}
Munro$\succ$Robertson$\succ$Balfour$\succ$Campbell-Sinclair

$\succ$Lamont$\succ$MacLeod$\succ$Wakeling$\succ$Cullen
\end{center}

\item \textbf{Highland, 2021 by-election, Ward 12 (Aird and Loch Ness)}

Winner: Fraser (Ind)

If we add 266 ballots of the form below then Shanks (SNP) becomes the winner. (Depending on how ties are broken, we could widen the range of ballots to 265-267.)

\begin{center}
Berkenheger$\succ$Fraser$\succ$MacKintosh$\succ$Moore$\succ$Robertson$\succ$Shanks
\end{center}

\item \textbf{Highland, 2022, Inverness West}

Winners: Boyd (SNP), Graham (LD), MacKintosh (Grn)

If we add 7-205 ballots of the form below then Fraser (Lab) replaces MacKintosh in the winner set.

\begin{center}
Boyd$\succ$MacKintosh$\succ$Graham$\succ$Forbes$\succ$Forsyth

$\succ$McDonald$\succ$Sansum$\succ$Smith$\succ$Fraser
\end{center}

\item \textbf{Inverclyde, 2017, Ward 5 (Inverclyde West)}


Winners: Ahlfeld (Ind), McEleny (SNP), Quinn (Ind)

If we add 99-178 ballots of the form below then Holliday (Lab) replaces Quinn in the winner set.

\begin{center}
Kelly$\succ$Quinn$\succ$Ahlfeld$\succ$McEleny$\succ$Taylor$\succ$White$\succ$Holliday
\end{center}

\item \textbf{Moray, 2012, Ward 5 (Heldon and Laich Ward)}


Winners: McGillivray (Ind), Ralph (SNP), Wright (Con), Tuke (Ind)

If we add 43-1058 ballots of the form below then Stewart (SNP) replaces Wright in the winner set.

If we add 1059-3296 ballots of the form below then Stewart replaces Tuke in the winner set.

\begin{center}
Ralph$\succ$McGillivray$\succ$Wright$\succ$Tuke$\succ$Gordon$\succ$Mackessack-Leitch$\succ$Stewart
\end{center}

\item \textbf{Moray, 2012, Ward 6 (Elgin City North)}


Winners: Gowans (SNP), Jarvis (Lab),  Shand (SNP)

We we add 253-336 ballots of the form below then Brown (Con) replaces Gowans in the winner set.

\begin{center}
Jarvis$\succ$Shand$\succ$Gowans$\succ$Margach$\succ$Brown
\end{center}

\item \textbf{Moray, 2022, Ward 4 (Fochabers Lhanbryde)}

Winners: MacRae (Con), Morrison (SNP), Williams (Lab)

If we add 22-1064 ballots of the form below then Bremner (SNP) replaces Williams in the winner set.

\begin{center}
Morrison$\succ$Williams$\succ$McRae$\succ$Cameron$\succ$Bremner
\end{center}

\item \textbf{North Ayrshire, 2017, Ward 1 (Irvine West)}

Winners: Clarkson (Lab), Gallacher (Con), MacAulay (SNP), McPhater (SNP)

If we add 17-2721 ballots of the form below then Limonci (SNP) replaces McPhater in the winner set.

\begin{center}
MacAulay$\succ$McPhater$\succ$Clarkson$\succ$Gallacher$\succ$Cochrane$\succ$Craig$\succ$Limonci
\end{center}

\item \textbf{North Ayrshire, 2017, Ward 5 (Ardrossan and Arran)}

Winners: Billings (Con), Gurney (SNP), McMaster (SNP)

If we add 497-994 ballots of the form below then McGuire (Lab) replaces McMaster in the winner set.

\begin{center}
Billings$\succ$Gurney$\succ$McMaster$\succ$Allison$\succ$Hunter$\succ$Turbett$\succ$McGuire
\end{center}

\item \textbf{North Ayrshire, 2017, Ward 9 (Saltcoats)}


Winners: McClung (SNP), McNicol (Ind), Montgomerie (Lab)

If we add 39-191 ballots of the form below then Clydesdale (Con) replaces McNicol in the winner set.

\begin{center}
Bianchini$\succ$McNicol$\succ$McClung$\succ$Montgomerie$\succ$Reid$\succ$Santos$\succ$Clydesdale
\end{center}

\item \textbf{North Ayrshire, 2022, Irvine West}

Winners: Gallacher (Con),  MacAulay (SNP), McPhater (Lab), Robertson (SNP)

If we add 363-1332 ballots of the form below then Mallinson (Lab) replaces McPhater in the winner set.

\begin{center}
Gallacher$\succ$Robertson$\succ$McPhater$\succ$MacAulay$\succ$Blades$\succ$

Cochrane$\succ$Hutton$\succ$Lindsay$\succ$Turbett$\succ$Mallinson
\end{center}

\item \textbf{North Lanarkshire, 2017, Ward 11 (Coatbridge South)}

Winners: Carragher (SNP), Castles (Lab), MacGregor (SNP), Encinias (Lab)

If we add 3-527 votes of the form below then Brooks (IANL) replaces Encinias in the winner set.

\begin{center}
Carragher$\succ$Castles$\succ$MacGregor$\succ$Encinias$\succ$Higgins$\succ$Cameron$\succ$Somers$\succ$Brooks
\end{center}

\item \textbf{North Lanarkshire, 2017, Ward 16 (Mossend and Holytown)}

Winners: Baird (SNP), McNally (SNP), Reddin (Lab)

If we add 22-147 ballots of the form below then Cunningham replaces Reddin in the winner set.

\begin{center}
Baird$\succ$Reddin$\succ$McNally$\succ$Clarkson$\succ$Cunningham
\end{center}

\item \textbf{North Lanarkshire, 2017, Ward 21 (Wishaw)}

Winners: Burgess (Con), Feeney (Lab), Fotheringham (SNP), Hume (SNP)

If we add 393-1521 ballots of the form below then Love (IANL) replaces Hume in the winner set.

\begin{center}
Burgess$\succ$Fotheringham$\succ$Feeney$\succ$Hume$\succ$McKay$\succ$Robertson$\succ$Love
\end{center}

\item \textbf{North Lanarkshire, 2022, Ward 12 (Airdrie South)}


Winners: Coyle (SNP), Di Mascio (SNP), McBride (Lab), Watson (Con)

If we add 22-461 ballots of the form below then McNeil (Lab) replaces Watson in the winner set.

\begin{center}
McBride$\succ$Watson$\succ$Coyle$\succ$Di Mascio$\succ$McNeil
\end{center}

\item \textbf{North Lanarkshire, 2022, Ward 16 (Mossend and Holytown)}


Winners: Baudo (SNP), McNally (Lab), Reddin (Lab)

If we add 28-192 ballots of the form below then Clarkson (SNP) replaces Reddin in the winner set.

\begin{center}
Cameron$\succ$Reddin$\succ$Baudo$\succ$McNally$\succ$Marshall$\succ$Clarkson
\end{center}

\item \textbf{North Lanarkshire, 2022, Ward 17 (Motherwell West)}

Winners: Crichton (SNP), Kelly (Lab), Nolan (Con)

If we add 5-889 ballots of the form below then Evans (SNP) replaces Nolan in the winner set.

\begin{center}
Crichton$\succ$Nolan$\succ$Kelly$\succ$McCann$\succ$Miller$\succ$Evans
\end{center}

\item \textbf{Perth Kinross, 2012, Ward 9 (Almond and Earn)}

Winners: Anderson (SNP), Jack (Ind), Livingstone (Con)

If we add 36-1668 ballots of the form below then Lumsden (SNP) replaces Jack in the winner set.

\begin{center}
Anderson$\succ$Livingstone$\succ$Jack$\succ$Dundas$\succ$Hayton$\succ$Lumsden
\end{center}

\item \textbf{Perth Kinross, 2017, Ward 10 (Perth City South)}

Winners: Band (SNP), Jamieson (Con), McCole (SNP), Wilson (LD)

If we add 621-783 ballots of the form below then Munro (Lab) replaces McCole in the winner set.

\begin{center}
Jamieson$\succ$Band$\succ$McCole$\succ$Wilson$\succ$Bathgate$\succ$Houston$\succ$Vallot$\succ$Munro
\end{center}

\item \textbf{Perth Kinross, 2022, Ward 4 (Highland)}

Winners: Duff (Con), McDade (Ind), Williamson (SNP)

If we add 56-73 ballots of the form below then Murray (SNP) replaces Williamson in the winner set.

\begin{center}
McDade$\succ$Williamson$\succ$Duff$\succ$Hunter$\succ$McDougall$\succ$McMahon$\succ$Metcalfe$\succ$Murray\\
\end{center}

\item \textbf{Renfrewshire, 2017, Ward 6 (Paisley Southeast)}


Winners: Devine (Lab), Mack (Ind), McGurk (SNP)

If we add 17-93 ballots of the form below then Fulton (Con) replaces Mack in the winner set.

\begin{center}
Swanson$\succ$Mack$\succ$Devine$\succ$McGurk$\succ$McShane

$\succ$Miller$\succ$Smith$\succ$Wilson$\succ$Fulton
\end{center}

\item \textbf{Renfrewshire, 2017, Ward 10 (Houston, Crosslee and Linwood)}

Winners: Doig (SNP), Dowling (Lab), Kerr (Con), Sheridan (Lab)

If we add 41-407 ballots of the form below then Innes (SNP) replaces Dowling in the winner set.

If we add 408-2100 ballots of the form below then Innes replaces Sheridan in the winner set.

\begin{center}
Doig$\succ$Dowling$\succ$kerr$\succ$Sheridan$\succ$Heron$\succ$Speirs$\succ$Innes
\end{center}

\item \textbf{South Ayrshire, 2012, Ward 3 (Ayr North)}

Winners: Campbell (SNP), Cavana (Lab), Hampton (Con), Miller (Lab)

If we add 120-1658 ballots of the form below then Slider (SNP) replaces Hampton in the winner set.

\begin{center}
Campbell$\succ$Miller$\succ$Hampton$\succ$Cavana$\succ$Slider
\end{center}

\item \textbf{South Ayrshire, 2012, Ward 4 (Ayr East)}

Winners: Douglas (SNP), Kilpatrick (Con), McGinley (Lab), Wilson (SNP)

If we add 49-1089 ballots of the form below then McKeand (Con) replaces Wilson in the winner set.

\begin{center}
Kilpatrick$\succ$Wilson$\succ$McGinley$\succ$Douglas$\succ$Bryden$\succ$Bulik$\succ$McKeand
\end{center}

\item \textbf{South Lanarkshire, 2012, Ward 3 (Clydesdale East)}

Winners: Barker (Lab), Gauld (SNP), Stewart (Con)

If we add 33-92 ballots of the form below then McAllan (SNP) replaces Barker in the winner set.

\begin{center}
Gauld$\succ$Barker$\succ$Stewart$\succ$McLatchie$\succ$Moxley$\succ$McAllan
\end{center}

\item \textbf{South Lanarkshire, 2012, Ward 10 (East Kilbride East)}

Winners: Cairney (Lab), Miller (SNP), Wardhaugh (SNP)

If we add 66-233 ballots of the form below then Scott (Lab) replaces Miller in the winner set.

If we add 234-1683 ballots of the form below then Scott replaces Wardhaugh in the winner set.

\begin{center}
Cairney$\succ$Miller$\succ$Wardhaugh$\succ$Benzies$\succ$Harrow$\succ$Scott
\end{center}

\item \textbf{South Lanarkshire, 2017, Ward 10 (East Kilbride East)}

Winners: Miller (SNP), Scott (Lab), Wardhaugh (SNP)

If we add 128-427 ballots of the form below then Peratt (Con) replaces Wardhaugh in the winner set.

Scott$\succ$Wardhaugh$\succ$Miller$\succ$Cairney$\succ$Doolan$\succ$Gall$\succ$Robb$\succ$Perratt \\

\item \textbf{South Lanarkshire, 2022, Ward 12 (Rutherglen Central and North)}

Winners: Calikes (SNP), Cowan (SNP), Lennon (Lab)

If we add 4-59 ballots of the form below then McGinty (Lab) replaces Cowan in the winner set.

\begin{center}
Cowan$\succ$Calikes$\succ$Lennon$\succ$Adebo$\succ$Fox$\succ$McRae$\succ$McGinty
\end{center}

\item \textbf{Stirling, 2017, Ward 3 (Dunblane and Bridge of Allan)}


Winners: Dodds (Con), Houston (SNP), Majury (Con), Tollemache (Grn)

If we add 136-365 ballots of the form below then Robbins (Lab) replaces Tollemache in the winner set.

\begin{center}
Hunter$\succ$Tollemache$\succ$Dodds$\succ$Houston$\succ$Majury$\succ$Auld$\succ$Robbins
\end{center}

\item \textbf{Stirling, 2022, Ward 4 (Stirling North)}

Winners: Gibson (Lab), McGill (SNP), Nunn (Con), Thomson (SNP)

If we add 279-1498 ballots of the form below then Smith (Grn) replaces Thomson in the winner set.

\begin{center}
Nunn$\succ$Gibson$\succ$Thomson$\succ$McGill$\succ$Franklin$\succ$McLelland$\succ$Smith
\end{center}

\item \textbf{West Dunbartonshire, 2017, Ward 2 (Leven)}

Winners: Bollan (WDuns), Dickson (SNP), McAllister (SNP), Millar (Lab)

If we add 17-1333 ballots of the form below then McGinty (Lab) replaces McAllister in the winner set.\\

\begin{center}
Millar$\succ$McAllister$\succ$Dickson$\succ$Bollan$\succ$Parlane$\succ$Quinn$\succ$Drummond$\succ$McGinty
\end{center}

\item \textbf{West Dunbartonshire, 2017, Ward 3 (Dumbarton)}

Winners: Conaghan (SNP), McBride (Lab), McLaren (SNP), Walker (Con)

If we add 373-1007 ballots of the form below then Black (WDuns) replaces McLaren in the winner set.

\begin{center}
Walker$\succ$McLaren$\succ$McBride$\succ$Conaghan$\succ$Muir$\succ$Ruine$\succ$Black
\end{center}

\item \textbf{West Lothian, 2022, Ward 7 (Whitburn and Blackburn)}

Winners: J. Dickson (SNP), M. Dickson (SNP), Paul (Lab), Sullivan (Lab)

If we add 49-520 ballots of the form below then Fairbairn replaces M. Dickson in the winner set.

\begin{center}
Sullivan$\succ$M. Dickson$\succ$J. Dickson$\succ$Paul$\succ$Pattle$\succ$Racionzer$\succ$Fairbairn
\end{center}

\end{itemize}

\end{document}